\documentclass[a4paper,11pt]{article}
\usepackage[margin=3cm]{geometry}
\date{April 20, 2026}

\usepackage[utf8]{inputenc}
\usepackage[T1]{fontenc}
\usepackage{soulutf8}
\sethlcolor{yellow}
\usepackage{authblk}

\usepackage{mathtools, amsthm, amsfonts, amssymb}
\usepackage{amsmath}
\usepackage{mathrsfs}
\usepackage{enumitem}
\usepackage{titlesec}
\usepackage{graphicx}
\usepackage{mathabx}
\usepackage{esint}
\usepackage{bbm}
\usepackage{bigints}
\usepackage[dvipsnames]{xcolor}
\usepackage[toc,page]{appendix}
\usepackage{comment}
\usepackage{multirow}

\usepackage{tikz}
\usetikzlibrary{decorations.pathreplacing}

\numberwithin{equation}{section}

\DeclareMathOperator*{\fiint}{\ensuremath{\iint\text{\kern-1.36em{\raisebox{5.87pt}{\rotatebox{-93}{$\setminus$}}}}}}

\allowdisplaybreaks

\usepackage[colorlinks=true,linkcolor=blue,citecolor=blue,urlcolor=blue,breaklinks]{hyperref}

\newtheorem{theorem}{Theorem}[section]

\newtheorem{lemma}[theorem]{Lemma}

\definecolor{darkgreen}{rgb}{0.0, 0.5, 0.0}

\def\tr{{\rm tr \,}}

\def\N{{\mathbb N}}
\def\Z{{\mathbb Z}}
\def\R{{\mathbb R}}
\def\C{{\mathbb C}}

\def\1{{\mathds{1}}}

\def\cE{{\mathcal E}}

\def\cB{{\mathcal B}}

\def\bR{{\mathbf R}}

\newcommand*{\Tr}{\mathrm{tr}}

\newcommand*{\rank}{\mathrm{Rank}}

\newcommand{\Supp}{{\mathrm{Supp}\,}}

\DeclarePairedDelimiterX\Set[1]\{\}{
   
   #1
}

\newcommand{\interior}[1]{
 {\kern0pt#1}^{\mathrm{o}}
}

\newlist{primenumerate}{enumerate}{1}
\setlist[primenumerate,1]{label={(H\arabic*$'$})}
\title{Eigenvalue asymptotics of Müller minimizers\\ for atoms and molecules}

\author[1,2,3]{Rupert L. Frank\thanks{\texttt{Emails: r.frank@lmu.de, longmeng@zju.edu.cn, nam@math.lmu.de, h.s@lmu.de}}}
\affil[1]{Institute of Mathematics, Ludwig Maximilian University Munich, Theresienstra\ss e 39, 80333 M\"unchen, Germany}
\affil[2]{Munich Center for Quantum Science and Technology, Schellingstra\ss e 4, 80799 M\"unchen, Germany}
\affil[3]{Department of Mathematics, Caltech, Pasadena, CA 91125, USA}
\author[4]{Long Meng}
\affil[4]{Center for Interdisciplinary Applied Mathematics \& Institute of Fundamental and
Transdiciplinary Research, Zhejiang University, China} 
\author[1,2]{Phan Thành Nam}
\author[1,2]{Heinz Siedentop}

\begin{document}

\maketitle

\abstract{
We study the spectral properties of minimizers of the M\"uller functional for atoms and molecules with $N$ electrons and total nuclear charge $Z$. We prove that under some suitable assumptions on $Z$ and $N$, the $k$-th eigenvalue of a M\"uller minimizer $\gamma_*$ behaves as $A_* k^{-8/3}$ when $k\to \infty$, with a constant $A_*>0$ determined explicitly by the density of $\gamma_*$. In particular, in the atomic case $V=Z|x|^{-1}$ our assumption holds if $Z$ is sufficiently large and $N\le Z- C_0 Z^{1/3}$. While our proof is inspired by Sobolev's work on the asymptotic behavior of the one-particle density matrix of Schr\"odinger ground states \cite{sobolev2022eigenvalue2}, the analysis in M\"uller theory requires several new ingredients concerning both the singular behavior of the integral kernel of the minimizers near the diagonal and the decay properties at infinity.}

\section{Introduction and main result}

M\"uller \cite{Mueller-84} (see also \cite{BuiBae-02}) has introduced an energy functional of the \emph{one-particle density matrix} (1-pdm) that has been widely used in computational quantum chemistry as an alternative to Hartree--Fock theory with an improved treatment of the exchange contribution. In this setting, $N$-electron systems with spin $q\in \mathbb{N}$ are described by self-adjoint operators $\gamma$ on $L^2(\R^3:\mathbb{C}^q)$ satisfying $0\leq\gamma\leq 1$ and $\Tr\gamma =N$. They are called \emph{one-particle density matrices}. To each $\gamma$ we can associate an integral kernel $\gamma(x,y)$ and a density $\rho_{\gamma}(x)=\Tr_{\C^q}[\gamma(x,x)]$ (the latter formula is formal but can be made precise by spectral decomposition, as explained below). 

The corresponding energy functional for a molecule with $N$ electrons, each of charge $-1$, and $K$ nuclei with fixed locations $\{\bR_j\}_{j=1}^K \subset \R^3$ and charges  $\{Z_j\}_{j=1}^K \subset (0,\infty)$, in units such that $\hbar^2/m=1$, is given by 
\begin{equation}
    \cE^{\rm M}(\gamma) := \tr(-\tfrac12\Delta \gamma)-\tr(V\gamma)+D(\rho_{\gamma},\rho_{\gamma})-X(\gamma^{1/2})
\end{equation}
with  the Coulomb potential
\begin{equation*}
    V(x):=\sum_{j=1}^K \frac{Z_j}{|x-\bR_j|},
\end{equation*}
the direct energy
\begin{equation*}
    D(\rho,\mu):=\frac{1}{2}\int_{\R^3}\int_{\R^3} \frac{\overline{\rho(x)}\mu(y)}{|x-y|}dxdy
\end{equation*}
and the exchange term 
\begin{equation*}
    X(\gamma):=\frac{1}{2}\int_{\R^3}\int_{\R^3}\frac{|\gamma (x,y)|^2}{|x-y|}dxdy.
\end{equation*}
The ground state energy in M\"uller theory is given by
\begin{equation}
\label{eq:muellermin}
E^{\rm M}(N):= \inf \left\{ \cE^{\rm M}(\gamma)\,|\, 0\le \gamma\le 1\text{ on } L^2(\R^3:\C^q),\, \Tr \gamma=N \right\}.
\end{equation}
The number $N$ stands for the number of electrons, but from a mathematical point of view the minimization problem defining $E^{\rm M}(N)$ makes sense for any positive real number $N$, and our results are valid in this generality.

Note that the only difference between the M\"uller energy functional $\cE^{\rm M}(\gamma)$ and the standard Hartree--Fock energy functional is that the standard exchange term $X(\gamma)$ is replaced by $X(\gamma^{1/2})$, where $\gamma^{1/2}$ is the (operator) square root of $\gamma$. Like the Hartree--Fock energy, the ground state energy $E^{\rm M}(N)$ also captures the Thomas--Fermi, Scott, and Dirac--Schwinger contributions \cite{Siedentop-09} for large $N$ and $Z$. But unlike the Hartree--Fock energy, the mapping $\gamma \mapsto \cE^{\rm M}(\gamma)$ is convex, leading to many remarkable properties that are not available in Hartree--Fock theory, including the uniqueness of the density $\rho_\gamma$ when a M\"uller minimizer $\gamma$ exists \cite{FraLieSeiSie-07}.

Concerning the existence of M\"uller miminizers, it was shown in \cite{FraLieSeiSie-07} that the minimization problem $E^{\rm M}(N)$ has a minimizer if $N\le Z=\sum_{j=1}^K Z_j$. On the other hand, in the atomic case $K=1$, it was proved in \cite{FraNamBos-18,FraNamBos-18b} that there is a critical electron number $N_c(Z) \le Z+ C$, with a universal constant $C>0$, such that \eqref{eq:muellermin} has no minimizers if $N>N_c(Z)$ (see also related results in the Hartree--Fock theory  \cite{Solovej-03} and the Thomas--Fermi--Dirac--Weizs\"acker theory \cite{FraNamVdB-18}).

In the present paper, we are interested in the eigenvalues of M\"uller minimizers. A crucial difference between M\"uller theory and Hartree--Fock theory is that every M\"uller minimizer $\gamma_*$ has infinitely many non-zero eigenvalues \cite{FraLieSeiSie-07}. To be precise, it was shown in \cite[Proposition 7]{FraLieSeiSie-07} that ${\rm Ran}(\tr_{\C^q}\gamma)=L^2(M_{\gamma_*})$ with $M_{\gamma_*}:=\{\rho_{\gamma_*} \neq 0\}$ and, clearly, $L^2(M_{\gamma_*})$ is infinite dimensional.

Our goal in this paper is to establish an asymptotic formula for the eigenvalues of M\"uller minimizers. Our main result provides the asymptotic behavior $A_* k^{-8/3}$ with a constant $A_*>0$ given explicitly in terms of the density $\rho_{\gamma_*}$. Remarkably, the decay exponenent $-\frac 83$ is the same as in the case of the 1-pdm of atomic ground states of the full Schr\"odinger Hamiltonian, thus reflecting the true quantum behavior. This $k^{-3/8}$-law in Schr\"odinger theory has been established recently by Sobolev \cite{sobolev2022eigenvalue2}, motivated by questions raised in quantum chemistry; see, e.g., \cite{cioslowski2019universalities,cioslowski2020off}. The constant $A_*$ is not universal and depends on the details of the system. The fact that it is strictly positive may be contrasted with the Schr\"odinger case  \cite{sobolev2022eigenvalue2}, where the constant may vanish in certain exceptional cases.

The precise statement of our main result is as follows.

\begin{theorem}[Eigenvalue asymptotics]\label{th:asymptotic} 
Consider the molecular case with total nuclear charge $Z=\sum_{j=1}^K Z_j$. Let $\gamma_*=\gamma_*(N)$ be a minimizer for the minimization problem in \eqref{eq:muellermin} and let 
\begin{equation}\label{eq:MF-operator-1body}
    k_{\gamma_*}:=-\frac{1}{2}\Delta- V+\rho_{\gamma_*(N)}*|\cdot|^{-1}
    \quad\text{in}\ L^2(\R^3:\C) \,.
\end{equation}
 We assume that $N\leq Z$ is chosen such that $k_{\gamma_*}$ has at least $N$ eigenvalues below $-1/2$, counted with multiplicity.  
Then the eigenvalues $1 \ge \lambda_1 \ge  \lambda_2 \ge.... \ge 0$ of  $\gamma_*$, counted with multiplicity, satisfy the asymptotic formula 
   \begin{equation} \label{eq:ev-asymptotics}
  \lim_{k\to \infty}k   \lambda_k^{3/8}= \frac{\sqrt{2}}{3 \pi^{5/4}} \int_{\R^3} \rho_{\gamma_*}(x)^{3/4}dx.
\end{equation}
In particular, \eqref{eq:ev-asymptotics} holds in the atomic case $V(x)=Z|x|^{-1}$ provided that $N\le Z - C_0 Z^{1/3}$ with a constant $C_0>(3/q)^{1/3}$ and $Z$ sufficiently large. 
\end{theorem}

The precise statement in the atomic case is that for every $C_0>(3/q)^{1/3}$ there is a $Z_C>0$ such that the result holds for all $Z\geq Z_C$ and $N\leq Z - C_0 Z^{1/3}$.

As can be seen from the proof, the two conditions $V(x)=Z|x|^{-1}$ and $N\le Z - C_0 Z^{1/3}$ in the atomic case  are both used to prove the spectral condition for $k_{\gamma_*}$, which in turn ensures  the exponential decay $e^{2\kappa |\cdot|} \rho_{\gamma_*} \in L^1(\R^3)$ for a constant $\kappa>0$. If we can derive this decay property by other means, then the asymptotic formula \eqref{eq:ev-asymptotics} remains correct even in the molecular case.

In the remainder of this introduction we explain the main ingredients in the proof of Theorem \ref{th:asymptotic}. In order to understand the eigenvalues of $\gamma_*$, we focus on the spectral properties of $\gamma_*^{1/2}$, which solves a suitable Euler--Lagrange equation \cite{FraLieSeiSie-07}. Inspired by Sobolev's argument \cite{sobolev2022eigenvalue2}, the eigenvalue asymptotics of $\gamma_*^{1/2}$ will be derived from a detailed description of the non-smooth behavior of the kernel
$$
\Phi(x,y)=(\gamma_*)^{1/2}(x,y)
$$ 
near the diagonal ${x=y}$, together with the exponential decay of $\Phi$ at infinity. We will discuss these two issues separately in the following.


\subsection*{Regularity of M\"uller minimizers}

By the spectral decomposition, every minimizer M\"uller $\gamma_*$ can be rewritten as
\begin{equation}\label{eq:gamma-dec}
    \gamma_*=\sum_{j=1}^\infty \lambda_j\left|\psi_j\right>\left<\psi_j\right|
\end{equation}
with the eigenfunctions $(\psi_j)_{j\geq 1}$ being an orthonormal basis for $L^2(\R^3:\C^q)$, and with the eigenvalues $1\ge \lambda_1\ge \lambda_2\ge ... \ge 0$ satisfying  $\sum_j \lambda_j=\Tr \gamma=N$. Then the density $\rho_{\gamma_*}$ can be defined properly as 
$$
\rho_{\gamma_*} (x) = \sum_{j=1}^\infty \lambda_j |\psi_j(x)|^2 \,.
$$
We note that the integral kernel of $\gamma_*$ is given by 
\begin{equation*}
\gamma_*(x,y)=\sum_{j=1}^\infty \lambda_j\psi_j(x)\psi_j^*(y)\in \C^{q\times q}
\end{equation*}
and, similarly, the integral kernel of the operator $\gamma_*^{1/2}$ is given by 
\begin{equation*}
      \Phi(x,y) :=  (\gamma_*)^{1/2}(x,y)=\sum_{j=1}^\infty \lambda_j^{1/2}\psi_j(x)\psi_j^*(y)\in \C^{q\times q} \,.
\end{equation*}
Whenever a M\"uller minimizer $\gamma_*$ exists, the corresponding function $\Phi(x,y)=(\gamma_*)^{1/2}(x,y)$ satisfies the following Euler--Lagrange equation \cite{FraLieSeiSie-07}
    \begin{equation}\label{eq:two-body}
     ( H_{\gamma_*}-2\mu)\Phi(x,y)=\sum_{\lambda_j=1}2e_j \psi_j(x)\psi_j^*(y)
    \end{equation}
with the two-body mean-field operator  
\begin{equation}\label{eq:MF-operator-2body}
    H_{\gamma_*}:=  -\frac{1}{2}\Delta_x -\frac{1}{2}\Delta_y-\phi_{\gamma_*}(x)-\phi_{\gamma_*}(y)-\frac{1}{|x-y|}
\end{equation}
and the screened potential 
$$
\phi_\gamma(x) :=V(x)-\int_{\R^3}\rho_{\gamma}(y)|x-y|^{-1}dy.
$$
In the two-body equation \eqref{eq:two-body}, the chemical potential $\mu$ on the left-hand side satisfies $\mu\leq -\frac{1}{8}$, while the coefficients $e_j$ on the right-hand side obey $e_j\leq 0$ for all $j$ and $e_j=0$ if $\lambda_j<1$. 

A crucial ingredient in the proof of Theorem \ref{th:asymptotic} will be a precise regularity result for $\Phi$. We note that the potential that enters the equation for $\Phi$ is singular on the sets 
\begin{equation}
    \label{eq:singularset}
    \{ (x,x):\ x\in\R^3\} \cup \bigcup_{j=1}^K \left(\left(\{\mathbf R_j\}\times\R^3 \right) \cup \left(\R^3\times\{\mathbf R_j\} \right) \right).
\end{equation}
These singularities will lead to non-smooth behavior of $\Phi$ on this set and it is precisely this non-smooth behavior of $\Phi$ that will lead to the rather slow decay of the eigenvalues of the corresponding operator $\gamma_*^{1/2}$. More precisely, we will see that the non-smoothness near the diagonal $\{(x,x): x\in\R^3\}$ is  the most relevant. 

One of our results concerns the regularity of $\Phi$ in the scale of Sobolev spaces $H^s(\R^d:\C^{q\times q})$. A more precise result, which will be crucial in our proof of Theorem \ref{th:asymptotic}, concerns the regularity in the scale of Besov spaces $B_{2,\infty}^s(\R^d:\C^{q\times q})$. Let us recall their definition. We set $(\Delta_{h} u)(z):= u(z+h)-u(z)$ for $z\in \R^d$ and $h\in \R^d$. Then for $l\in \mathbb{N}$,
\begin{equation}\label{eq:Delta-h-l}
\begin{split}
     \Delta_{h}^{(l)}u(z)&= \sum_{j=0}^l (-1)^{j+1}\begin{pmatrix}
        l\\j
    \end{pmatrix}u(z+jh)\\
    &=\int_0^1\int_0^1\cdots\int_0^1 (h\cdot \nabla)^l u(z+\sum_{j=1}^l s_jh) ds_1ds_2\cdots ds_l.
\end{split}
\end{equation}
The Besov space $B_{2,\infty}^s(\R^d:\C^{q\times q})$ for $s>0$ consists of $u\in L^2(\R^d:\C^{q\times q})$ such that for some integer $l>s$ we have
\begin{equation*}
  \|u\|_{B_{2,\infty}^s(\R^d:\C^{q\times q})}= \|u\|_{L^2(\R^d:\C^{q\times q})}+  \sup_{h\neq 0}|h|^{-s}\|\Delta_h^{(l)} u\|_{L^2(\R^d:\C^{q\times q})}<\infty.
\end{equation*}
It is known that different choices of $l$ with $l>s$ lead to equivalent norms.

In addition to a sharp regularity result for $\Phi$, we show that one can `extract' from $\Phi$ an approximation to the non-smooth behavior near the singular set \eqref{eq:singularset} and obtain an improved regularity result. We take a function $ \theta\in C^\infty(\R:[0,1])$ satisfying 
\begin{equation}\label{eq:theta}
    \theta(t)=t \quad{\rm for } \quad t\in [-1/2,1/2],\qquad {\rm and} \qquad \theta=1 \quad{\rm for} \quad|t|\geq 1,
\end{equation}
and define the Jastrow factor by 
\begin{equation}\label{eq:Jastrow}
    F(x,y):=-\sum_{j=1}^KZ_j\theta(|x-\bR_j|)-\sum_{j=1}^KZ_j\theta(|y-\bR_j|)-\frac{1}{2}\theta(|x-y|) \,.
\end{equation}
Our result on the regularity of $\Phi$ is as follows.

\begin{theorem}[Two-body regularity]\label{th:regularity1} 
    Let $\gamma_*$ be a minimizer for \eqref{eq:muellermin} in the molecular case. Then $\Phi(x,y)=\gamma_*^{1/2}(x,y)$, the integral kernel of $\gamma_*^{1/2}$, satisfies 
$$\Phi\in H^{2+\alpha}(\R^6:\C^{q\times q}) \cap B^{5/2}_{2,\infty}(\R^6:\C^{q\times q})$$
for any $0\leq \alpha<\frac{1}{2} $. Moreover, the function $\Psi(x,y)=e^{-F(x,y)}\Phi(x,y)$, with $F$ the Jastrow factor in \eqref{eq:Jastrow}, satisfies the improved regularity
$$\Psi \in H^{3+\alpha}(\R^6:\C^{q\times q})\cap B^{7/2}_{2,\infty}(\R^6:\C^{q\times q})$$
for any $0\leq \alpha<\frac{1}{2}$.
\end{theorem}

The proof of this theorem is given in Section \ref{sec:regularity}.

We emphasize that Theorem \ref{th:regularity1} holds as soon as the M\"uller minimizer exists, without any additional assumption on $N$ or $Z= \sum_{j=1}^K Z_j$. The proof of Theorem \ref{th:regularity1} is based on the Euler--Lagrange equation \eqref{eq:two-body}. As we have already mentioned, the non-smoothness of $\Phi$ comes from the singularity of the potential $-\phi_{\gamma_*}(x) -\phi_{\gamma_*}(y)- |x-y|^{-1}$ of the two-body mean-field operator $H_{\gamma_*}$ introduced in \eqref{eq:MF-operator-2body} on the set \eqref{eq:singularset}. Hence, in some sense, the nature of the non-smoothness of $\Phi$ from M\"uller's equation \eqref{eq:two-body} is similar to that of the many-body Schr\"odinger equation \cite{sobolev2022eigenvalue2}. It is remarkable that the sign of the electron-electron interaction (which is negative in our case, while it is positive in the many-body Schr\"odinger equation) is irrelevant in this question.

The mathematical study of the Coulombic eigenfunctions goes back at least to Kato’s cusp condition \cite{Kato1957}. Using Jastrow factors, the regularity can be improved \cite{Soren-etal-SharpRegularity}; see also \cite{Soren-Ostergaard-derivatives} and references therein. Related results are available for the density and one-particle density matrix \cite{Peter-Sobolev-Analyticity,Soren-etal-density}. On the other hand, to circumvent the curse of dimension in numerical analysis, mixed regularity in a global Hilbert-space setting was proved in \cite{Yserentant3,Yserentant1} and further improved in \cite{meng2023mixed}.

In contrast to the above-mentioned works on regularity, we have to prove regularity of the M\"uller wavefunction in the scale of Besov spaces. To our knowledge, these regularity results in terms of Besov space are novel. In addition, our result is optimal in the sense that $\Phi$ (resp.~$\Psi$)  does not belong to $B^s_{2,\infty}(\R^6:\C^{q\times q})$ for $s>5/2$ (resp.~$s>7/2$) nor to $H^{t}(\R^{6})(\R^6:\C^{q\times q})$ for $t\geq 5/2$ (resp.~$t\geq 7/2$). We do not give a proof of these optimality relations, except for an indirect argument for the Besov-space regularity of $\Phi$. Indeed, if $\Phi$ would belong to $B^s_{2,\infty}(\R^6:\C^{q\times q})$ with $s>5/2$ and if $\rho_{\gamma_*}$ would satisfy the exponential decay in \eqref{eq:decayass}, then we could follow Step 1 in the proof of Theorem \ref{th:asymptotic} and deduce that $\gamma_*^{1/2}\in\mathfrak S_{6/(3+2s),\infty}$. This would mean that $\lambda_k = O(k^{-(3+2s)/3})$, which would contradict Theorem \ref{th:asymptotic} since the right side of \eqref{eq:ev-asymptotics} is strictly positive.


\subsection*{Decay of M\"uller minimizers}

Next, we consider the decay of $\rho_{\gamma_*}$, which turns out to be a serious challenge. It is difficult to use the two-body equation \eqref{eq:two-body} for two reasons:  (1) the chemical potential $\mu$ may lie in the essential spectrum of $H_{\gamma}$, or even at its threshold, which leads to the failure of the standard exponential decay argument; (2) \eqref{eq:two-body} is not an eigenvalue problem, and the right-hand side of \eqref{eq:two-body} complicates the analysis. In addition, we do not know the exact values of $e_j$ in \eqref{eq:two-body} when $\lambda_j=1$, which prevents bootstrap arguments. 

Therefore, it is natural to replace the two-body equation \eqref{eq:two-body} by one-body equations. It was shown in \cite{FraLieSeiSie-07} that \eqref{eq:two-body} implies the equations
\begin{equation}\label{eq:Muller-eigen}
    \begin{cases}
    h_{\gamma_*} \psi_j= \mu \psi_j & \text{if}\ \lambda_j<1\\
       h_{\gamma_*} \psi_j= (\mu+e_j) \psi_j & \text{if}\ \lambda_j=1.
\end{cases}
\end{equation}
with the same $\mu$ and $e_j$ as in \eqref{eq:two-body}. Here the one-body operator $h_{\gamma_*}$ is defined by
\begin{equation}\label{eq:onebodyh}
    h_{\gamma_*} :=-\frac{1}{2}\Delta-\phi_{\gamma_*} -\mathfrak{X}_{\gamma_*}
    \quad \text{in}\ L^2(\R^3:\C^q)
\end{equation}
with $\mathfrak{X}_{\gamma_*}$ the nonlocal exchange operator with matrix elements
\begin{equation*}
    \left<\psi_i,\mathfrak{X}_{\gamma_*} \psi_j\right>=(\sqrt{\lambda_i}+\sqrt{\lambda_j})^{-1}\left<\psi_i,Z_{\gamma_*}\psi_j\right>.
\end{equation*}
Equivalently, we can write 
\begin{equation}\label{eq:xchangeop}
    \mathfrak{X}_{\gamma_*}=\frac{1}{\pi}\int_0^\infty \frac{1}{\gamma+s}Z_{\gamma_*}\frac{1}{\gamma+s}\sqrt{s}ds
\end{equation}
with $Z_{\gamma_*}$ being the operator with $\C^q\times\C^q$-valued integral kernel
\begin{equation*}
    Z_{\gamma_*}(x,y):=\gamma_*^{1/2}(x,y) \, |x-y|^{-1}.
\end{equation*}

As $\rank(\gamma_*)=\infty$ and $\tr\gamma_*=N$, the number $|\{j\in \N; \lambda_j=1\}|\leq N-1$. Consequently, $\mu$ is an eigenvalue of $h_{\gamma_*}$ with infinite multiplicity. Unlike the standard Hartree--Fock exchange operator, $\mathfrak{X}_{\gamma_*}$ is not a bounded operator and may not even be densely defined. In fact, $\mathfrak{X}_{\gamma_*}$, consequently also $h_{\gamma_*}$, are only defined as unbounded operators on the orthogonal complement of $\ker(\gamma_*)$. (Indeed, for $\lambda_i=\lambda_j=0$ the matrix elements $\left<\psi_i,\mathfrak{X}_{\gamma_*} \psi_j\right>$ is not well-defined.)

In summary, without extra information on $\mu$, neither \eqref{eq:two-body} nor \eqref{eq:Muller-eigen} seems to be helpful for obtaining a good decay estimate for $\rho_{\gamma_*}$. Our key observation is that if $\mu<-1/2$, then such a decay can be obtained from these equations using the fact that $ \mathfrak{X}_{\gamma_*} \ge 0$ (the unboundedness is not an issue in this case). To be precise, we have the following conditional decay estimate. 

\begin{theorem}[Decay estimate]\label{th:decay} 
    Let $\gamma_*$ be a minimizer for \eqref{eq:muellermin} in the molecular case and assume that the chemical potential $\mu$ in \eqref{eq:two-body} satisfies $\mu<-\frac{1}{2}$. Then, for any $0<\kappa< \sqrt{-\frac{1}{2}-\mu}$, we have 
    \begin{equation*}
        \int_{\mathbb{R}^3} e^{2\kappa|x|}\rho_{\gamma_*}(x)dx<\infty.
    \end{equation*}
\end{theorem}

For the proof of Theorem \ref{th:decay} we need to change the standard argument used to get exponential decay; see e.g. \cite[Appendix B]{RooSei-22}. Essentially, we integrate the two-body equation \eqref{eq:two-body} against $e^{2 \kappa |x|}\overline{\Phi(x,y)}$, with a suitable cut-off in the exponential weight $|x|$. The extra condition $\mu<-\frac{1}{2}$ allows us to treat the negative potential $-|x-y|^{-1}$ in \eqref{eq:two-body} in a non-balanced manner, making it possible to compare with the corresponding outcome from the one-body equation \eqref{eq:Muller-eigen}. The details will be given in Section \ref{sec:decay}.

In the last step, we still need to justify whether the extra condition $\mu<-1/2$ holds true. Recall that in general we only know the bound $\mu\le-\frac 18$. The latter bound is deduced in \cite{FraLieSeiSie-07} from the fact that 
\begin{equation*}
    \inf {\rm spec} \left(-\tfrac{1}{2}\Delta_x -\tfrac{1}{2}|x-y|^{-1} \right)\leq -\tfrac{1}{8} \,.
\end{equation*}
In the following, we improve the bound for the chemical potential by using the one-body mean-field operator $k_{\gamma_*}$ introduced in \eqref{eq:MF-operator-1body}. 

\begin{theorem}[Improved bound for the chemical potential]\label{th:mu} 
    Consider the molecular case with total nuclear charge $Z=\sum_{j=1}^K Z_j$. Assume that there is a minimizer $\gamma_*$ for \eqref{eq:muellermin} and let $1\geq \lambda_1\geq \lambda_2\cdots\geq 0$ be its eigenvalues and let $J:=\inf\{j\in \N:\; \lambda_j<1\}$.

    For every self-adjoint operator $h$ in $L^2(\R^3:\C^q)$ with
    $h\geq k_{\gamma_*} \otimes 1_{\C^q}$, if $h$ has  at least $J$ eigenvalues $\sigma_1(h) \le \sigma_2(h) \le \cdots \le \sigma_J(h) $  below its essential spectrum, with eigenfunctions $u_1,\ldots,u_J$, then 
    \begin{equation}\label{eq:mu-sigmaN}
        \mu\leq \sigma_J(h)-\tfrac{1}{2}\inf_{\substack{\|u\|_{L^2}=1\\ u\in {\rm Span}\{u_1,\ldots,u_J\}}} \int_{\R^3}\int_{\R^3} \frac{|u|^2(x)|u|^2(y)}{|x-y|}dxdy<\sigma_J(h).
    \end{equation}
    Consequently, the bound $\mu<-1/2$ holds in the atomic case  $V(x)=Z |x|^{-1}$ provided that $N\leq Z-C_0 Z^{1/3}$  with a constant $C_0> (3/q)^{1/3}$ and $Z$ sufficiently large. 
\end{theorem}

The proof of Theorem \ref{th:mu} is based on the variational principle. More precisely, we use the upper bound 
\begin{equation*}
    \mu:=\mu(N) \le \lim_{t\to 0^+}\frac{E^{\rm M}(N+t)-E^{\rm M}(N)}{t}
\end{equation*}
from \cite{FraLieSeiSie-07}, together with a suitable trial state for $E^{\rm M}(N+t)$ to obtain \eqref{eq:mu-sigmaN}. In the atomic case, since $\rho_{\gamma_*}$ is radial, we can compare $k_{\gamma_*}$ with the Hamiltonian of the hydrogen atom, and then use the spectral properties of the latter to conclude the desired estimate $\mu<\sigma_N(h)\le -1/2$ when $N\le Z - C_0 Z^{1/3}$.

Finally, we will conclude the proof of Theorem \ref{th:asymptotic} in Section \ref{sec:main-thm}. This part of our argument is similar to Sobolev's analysis \cite{sobolev2022eigenvalue2} and is based on the spectral estimates and spectral asymptotics for certain integral operators.

\medskip

\noindent\textit{Acknowledgments.} This work was partially funded by the Deutsche Forschungsgemeinschaft (DFG, German Research Foundation) via the TRR 352 – Project-ID 470903074. 
Partial support through US National Science Foundation grant DMS-1954995 (R.L.F.), the German Research Foundation grant EXC-2111-390814868 (R.L.F., P.T.N.), and the European Research Council via the ERC Consolidator Grant RAMBAS Project-Nr. 10104424 (L.M., P.T.N.) is acknowledged.


\section{Proof of Theorem \ref{th:regularity1}: Regularity estimates} \label{sec:regularity}

In this section we prove Theorem \ref{th:regularity1}. The arguments are inspired, in part, by \cite{meng2023mixed}.

For notational convenience we focus on the atomic case, i.e., $V(x)=\frac{Z}{|x|}$ and the spinless case, i.e, $q=1$. But the proof works for the molecular case, i.e., $V= \sum_{j=1}^K \frac{Z_j}{|x-\bR_j|}$ and general spin $q\in \N^+$. For simplicity, we will also write $\gamma$ instead of $\gamma_*$.

We divide the analysis into four steps, occupying the corresponding subsections.

\subsection{Proof of $\Phi\in H^{2+\alpha}(\R^6)$}

We recall that we write $\Phi(x,y)=\gamma^{1/2}(x,y)$. We first show that $\Phi\in H^{2+\alpha}(\R^6)$ for any $0\leq\alpha<\frac{1}{2}$. As $\gamma$ is a minimizer of \eqref{eq:muellermin}, we have $\Tr \gamma =N$ and $\tr(\sqrt{-\Delta}\gamma \sqrt{-\Delta} ) <\infty$. Thus,
\begin{equation*}
\begin{split}
    \|\Phi\|^2_{H^1(\R^6)}&=\|\Phi\|_{L^2(\R^6)}^2+ \|\sqrt{-\Delta_x}\Phi\|_{L^2(\R^6)}^2+\|\sqrt{-\Delta_y}\Phi\|_{L^2(\R^6)}^2\\
    &=\tr(\gamma)+ 2\tr(\sqrt{-\Delta}\gamma \sqrt{-\Delta} )<\infty.
\end{split}
\end{equation*}
Then for any $j\geq 1$,
\begin{equation}\label{eq:psi_j}
    \lambda_j \|\psi_j\|_{H^1(\R^3)}=\|\gamma^{1/2}\psi_j\|_{H^1(\R^3)}\leq \|\Phi\|_{H^1(\R^6)}<+\infty.
\end{equation}

Consequently, for $j$ satisfying $\lambda_j=1$ we have
\begin{equation}
\begin{split}
     \|\psi_j(x)\psi_j^*(y)\|_{H^1(\R^6)}&\leq  \|\sqrt{1-\Delta_x}\psi_j(x)\psi_j^*(y)\|_{L^2(\R^6)}+\|\psi_j(x)\sqrt{1-\Delta_y}\psi_j^*(y)\|_{L^2(\R^6)}\\
     &\lesssim \|\psi_j\|_{H^1(\R^3)}.
\end{split}
\end{equation}
Moreover, using 
\begin{equation}
    |\{\lambda_j=1\}|= \sum_{\lambda_j=1} 1 \leq \sum_{j}\lambda_j= N \,,
\end{equation}
we deduce that $\lambda_j<1$ for $j\geq N+1$, and hence 
\begin{equation}
    \sup_{\lambda_j=1}|e_j|\leq \max_{j=1,\cdots,N} |e_j|<\infty.
\end{equation}
As a result,  we estimate the right-hand side of  \eqref{eq:two-body} as 
\begin{equation}\label{eq:right-finite}
    \left\|\sum_{\lambda_j=1}e_j \psi_j(x)\psi_j^*(y)\right\|_{H^1(\R^6)}\leq \sum_{\lambda_j=1}|e_j| \|\psi_j(x)\psi_j^*(y)\|_{H^1(\R^6)}<\infty .
\end{equation}

\medskip

To show $\Phi\in H^{2+\alpha}(\R^6)$ for any $\alpha<\frac{1}{2}$, we first show $\|\Phi\|_{H^2(\R^6)}<+\infty$. By \eqref{eq:two-body}, 
\begin{equation}
\begin{split}
     & \left\|(-\Delta_x-\Delta_y)\Phi\right\|_{L^2(\R^6)}\\
   \leq& 2\left\|\left(\phi_{\gamma}(x)+\phi_{\gamma}(y)+|x-y|^{-1}+ 2\mu\right)\Phi\right\|_{L^2(\R^6)}+\left\|\sum_{\lambda_j=1}e_j \psi_j(x)\psi_j^*(y)\right\|_{L^2(\R^6)}.
\end{split}
\end{equation}
By Hardy's inequality,
\begin{equation}
    \|\phi_{\gamma}(x)\Phi\|_{L^2(\R^6)}\leq 2(Z+N)\|\nabla_x\Phi\|_{L^2(\R^6)}\leq 2(Z+N)\|\Phi\|_{H^1(\R^6)},
\end{equation}
and
\begin{equation}
     \||x-y|^{-1}\Phi\|_{L^2(\R^6)}\leq 2\|\nabla_x\Phi\|_{L^2(\R^6)}\leq 2\|\Phi\|_{H^1(\R^6)}.
\end{equation}
These two inequalities and \eqref{eq:right-finite} yield
\begin{equation}
     \|(-\Delta_x-\Delta_y)\Phi\|_{L^2(\R^6)}<\infty.
\end{equation}
Thus,
\begin{equation}
    \|\Phi\|_{H^2(\R^6)}\leq \|\Phi\|_{L^2(\R^6)}+\|(-\Delta_x-\Delta_y)\Phi\|_{L^2(\R^6)}<\infty.
\end{equation}

\medskip

Now we are going to show that $\Phi\in H^{2+\alpha}$ for any $0\leq \alpha<\frac{1}{2}$. It suffices to show that
\begin{equation}
    \|(-\Delta_x-\Delta_y)^{1+\alpha/2} \Phi\|_{L^2(\R^6)}<\infty.
\end{equation}
By \eqref{eq:two-body} again,
\begin{equation}\label{eq:1.8}
    \begin{split}
       & \|(-\Delta_x-\Delta_y)^{1+\alpha/2} \Phi\|_{L^2(\R^6)}\\
       \leq& 2 \left\|(-\Delta_x-\Delta_y)^{\alpha/2} \left(\left(\phi_{\gamma}(x)+\phi_{\gamma}(y)+|x-y|^{-1}+ 2\mu\right)\Phi\right)\right\|_{L^2(\R^6)}+  \left\|\sum_{\lambda_j=1}e_j \psi_j(x)\psi_j^*(y)\right\|_{H^1(\R^6)}\\
       \leq& 2 \left\|(1-\Delta_x)^{\alpha/2}\left(\left(\phi_{\gamma}(x)+\phi_{\gamma}(y)+|x-y|^{-1}+ 2\mu\right)\Phi\right)\right\|_{L^2(\R^6)}\\
       &+ 2 \left\|(1-\Delta_y)^{\alpha/2}\left(\left(\phi_{\gamma}(x) +\phi_{\gamma}(y)+|x-y|^{-1}+2\mu\right)\Phi\right)\right\|+ \left\|\sum_{\lambda_j=1}e_j \psi_j(x)\psi_j^*(y)\right\|_{H^1(\R^6)}.
    \end{split}
\end{equation}

We begin by considering the first term on the right-hand side of \eqref{eq:1.8}. 
First, by Hardy's inequality,
\begin{equation}\label{eq:2.12a1}
     \left\|(1-\Delta_x)^{\alpha/2}\phi_{\gamma}(y)\Phi\right\|_{L^2(\R^6)}=    \left\|\phi_{\gamma}(y)(1-\Delta_x)^{\alpha/2}\Phi\right\|_{L^2(\R^6)}\lesssim \|\Phi\|_{H^2(\R^6)}.
\end{equation}
Next,
\begin{equation}\label{eq:2.12a2}
\begin{split}
     \left\|(1-\Delta_x)^{\alpha/2}\phi_{\gamma}(x)\Phi\right\|_{L^2(\R^6)}
       \leq &Z\left\|(1-\Delta_x)^{\alpha/2}(|x|^{-1}\Phi)\right\|_{L^2(\R^6)}\\
       &+\left\|(1-\Delta_x)^{\alpha/2}(\rho_{\gamma}*|\cdot|^{-1}(x)\Phi)\right\|_{L^2(\R^6)}.
\end{split}
\end{equation} 
In the molecular case, we only need to split $V$ in $\phi_\gamma$ and treat each nucleus separately.

Concerning the term $\rho_{\gamma}*|\cdot|^{-1}(x)$, by Hardy's inequality again, 
\begin{equation}\label{eq:hartree}
    \begin{split}
       & \|\rho_{\gamma}*|\cdot|^{-1}\|_{W^{1,\infty}(\R^3)}\leq \|\rho_{\gamma}*|\cdot|^{-1}\|_{L^\infty(\R^3)}+\|\nabla (\rho_{\gamma}*|\cdot|^{-1})\|_{L^\infty(\R^3)}\\
        \leq & \sup_{x\in \R^3}\int_{\R^3}\int_{\R^3}\frac{|\Phi(y,z)|^2}{|x-y|}dydz+\sup_{x\in \R^3}\int_{\R^3}\int_{\R^3}\frac{|\Phi(y,z)|^2}{|x-y|^2}dydz \lesssim \|\Phi\|_{H^1(\R^6)}^2.
    \end{split}
\end{equation}
Then
\begin{equation}\label{eq:1.15}
    \begin{split}
       &\left\|(1-\Delta_x)^{\alpha/2}(\rho_{\gamma}*|\cdot|^{-1}(x)\Phi)\right\|_{L^2(\R^6)}\leq \|\rho_{\gamma}*|\cdot|^{-1}(x)\Phi\|_{H^1(\R^6)}\\
       \leq & \|\rho_{\gamma}*|\cdot|^{-1}\|_{W^{1,\infty}(\R^3)}\|\Phi\|_{H^1(\R^6)}\lesssim\|\Phi\|_{H^1(\R^6)}^2.
    \end{split}
\end{equation}

Next, we study the term $Z_j|x|^{-1}\Phi$. Using $(1-\Delta_x)^{\alpha/2}=(1-\Delta_x)^{(\alpha-1)/2}(1-\Delta_x)^{1/2}$, we have for $\alpha<\frac{1}{2}$,
\begin{equation}\label{eq:1.13}
    \begin{split}
        &   \left\|(1-\Delta_x)^{\alpha/2}(|x|^{-1}\Phi)\right\|_{L^2(\R^6)}\\
 \leq&  \left\|(1-\Delta_x)^{(\alpha-1)/2}(|x|^{-1}\Phi)\right\|_{L^2(\R^6)}+\left\|\nabla_x (1-\Delta_x)^{(\alpha-1)/2}(|x|^{-1}\Phi)\right\|_{L^2(\R^6)}\\
 \leq&  \left\||x|^{-1}\Phi\right\|_{L^2(\R^6)}+\left\|\ (1-\Delta_x)^{(\alpha-1)/2}\left( \frac{x}{|x|^3}\Phi\right)\right\|_{L^2(\R^6)}\\
 &+\left\|\ (1-\Delta_x)^{(\alpha-1)/2}\left( |x|^{-1}\nabla_x\Phi\right)\right\|_{L^2(\R^6)}\\
 \lesssim& \|\Phi\|_{H^1(\R^6)}+ \| |x|^{-1-\alpha}\Phi\|_{L^2(\R^6)}+\||x|^{-\alpha}\nabla_x \Phi\|_{L^2(\R^6)}\\
 \lesssim& \|\Phi\|_{H^1(\R^6)}+\|(-\Delta)^{(1+\alpha)/2}\Phi\|_{L^2(\R^6)} \lesssim \|\Phi\|_{H^2(\R^6)}<\infty.
    \end{split}
\end{equation}
In the last two inequalities, we used Herbst's inequality \cite{Herbst77} twice: for any $0\leq s<\frac{3}{2}$
\begin{equation}
    \|(-\Delta)^{-s/2}|\cdot|^{-s}\|_{\cB(L^2(\R^3))}<+\infty.
\end{equation}
Note that the assumption $\alpha<1/2$ is dictated by the assumption $s<3/2$ in the last inequality with $s=1+\alpha$. Thus the term in \eqref{eq:2.12a2} can be estimated using the above bounds. 

Concerning the term $(1-\Delta_x)^{\alpha/2}\left(|x-y|^{-1}\Phi\right)$, we need to apply Herbst's inequality in a slightly different way: for any fixed $y\in \R^3$, any $f\in H^1(\R^6)$ and $0\leq s< \frac{3}{2}$, we have
\begin{equation}
    \int_{\R^3}\frac{|f(x,y)|^2}{|x-y|^{2s}}dx\lesssim  \int_{\R^3}|(-\Delta_x)^{s/2}f|^2(x,y)dx.
\end{equation}
Thus,
\begin{equation}
    \int_{\R^3} \int_{\R^3}\frac{|f(x,y)|^2}{|x-y|^{2s}}dxdy\lesssim   \int_{\R^3}\int_{\R^3}|(-\Delta_x)^{s/2}f|^2(x,y)dxdy.
\end{equation}
By duality, we also have
\begin{equation}
     \int_{\R^3} \int_{\R^3}|(-\Delta_x)^{-s/2} f|^2(x,y)dxdy\lesssim   \int_{\R^3}\int_{\R^3}||x-y|^{s}f|^2(x,y)dxdy.
\end{equation}
Then proceeding as for \eqref{eq:1.13}
\begin{equation}\label{eq:2.12a3}
  \begin{split}
        &   \left\|(1-\Delta_x)^{\alpha/2}(|x-y|^{-1}\Phi)\right\|_{L^2(\R^6)}\\
\leq&  \left\||x-y|^{-1}\Phi\right\|_{L^2(\R^6)}+\left\|\ (1-\Delta_x)^{(\alpha-1)/2}\left( \frac{x-y}{|x-y|^3}\Phi\right)\right\|_{L^2(\R^6)}\\
&+\left\|(1-\Delta_x)^{(\alpha-1)/2}\left( |x-y|^{-1}\nabla_x\Phi\right)\right\|_{L^2(\R^6)}\\
 \lesssim& \|\Phi\|_{H^1(\R^6)}+ \| |x-y|^{-1-\alpha}\Phi\|_{L^2(\R^6)}+\||x-y|^{-\alpha}\nabla_x \Phi\|_{L^2(\R^6)}\\
 \lesssim& \|\Phi\|_{H^1(\R^6)}+\|(-\Delta_x)^{(1+\alpha)/2}\Phi\|_{L^2(\R^6)} \leq \|\Phi\|_{H^2(\R^6)}
    \end{split}
\end{equation}

To summarize, the first term on right-hand side in \eqref{eq:1.8} can be bounded using \eqref{eq:2.12a1}, \eqref{eq:2.12a2} and \eqref{eq:2.12a3}. The study of the second term on the right-hand side in \eqref{eq:1.8} is similar and the third term is finite by \eqref{eq:right-finite}. Thus $\Phi\in H^{2+\alpha}(\R^6)$ for any $0\leq \alpha<\frac{1}{2}$.


\subsection{Proof of $\Phi\in B^{5/2}_{2,\infty}(\R^6)$}\label{sec:2.2}

Before proving $\Phi\in B^{5/2}_{2,\infty}(\R^6)$, we derive a useful inequality concerning difference quotients of the function $Q(x)=|x|^{-1}$, $x\in\R^3$. Writing $h_0= h/|h|\in\mathbb S^2$ and changing variables $x=|h|x'$, we find that
\begin{equation}\label{eq:h1/2}
\begin{split}
\left\|\Delta_{h} Q\right\|_{L^2(\R^3)}^2&=\int_{\R^3} \left||x+h_0 |h||^{-1}-|x|^{-1}\right|^2 dx\\
&=|h|\int_{\R^3} \left||x'+h_0|^{-1}-|x'|^{-1}\right|^2 dx'\lesssim |h|.
\end{split}
\end{equation}

Next, we observe the following useful fact.

\begin{lemma}
    Let $d\in\N^*$ and $s\geq 2$. If $u\in L^2(\R^d)$ and $\Delta u\in B^{s-2}_{2,\infty}(\R^d)$, then $u\in B^{s}_{2,\infty}(\R^d)$.
\end{lemma}
\begin{proof}
 By the second equation in \eqref{eq:Delta-h-l}, we have for any $d\in \N^*$,
\begin{equation}\label{eq:Delta-h-Delta}
       \|\Delta^{(l)}_{h} u\|_{L^2(\R^d)}= \|\Delta^{(2)}_{h} \Delta^{(l-2)}_{h} u\|_{L^2(\R^d)}\lesssim|h|^2\|\Delta^{(l-2)}_{h}\Delta u\|_{L^2(\R^d)}.
\end{equation}
This proves the lemma.
\end{proof}

Thus,  to show $\Phi\in B^{5/2}_{2,\infty}(\R^6)$, it suffices to show that $(-\Delta_x-\Delta_y) \Phi\in B^{1/2}_{2,\infty}(\R^6)$. By the equation \eqref{eq:two-body}, it suffices to show
\begin{equation}
    \frac{1}{|x|}\Phi,\frac{1}{|y|}\Phi,\frac{1}{|x-y|}\Phi \in B^{1/2}_{2,\infty}(\R^6)\label{eq:1.55}.
\end{equation}
The fact that other terms of \eqref{eq:two-body} are in $H^{1}(\R^6)\subset B_{2,\infty}^{1/2}(\R^6)$ follows immediately from \eqref{eq:right-finite} and \eqref{eq:1.15}.

It remains to prove \eqref{eq:1.55}. We will use the inequality
$$
\| \Delta_h (fg) \|_{L^2(\R^d)} \leq \| (\Delta_h f) g \|_{L^2(\R^d)} + \| f \Delta_{-h} g \|_{L^2(\R^d)} \,,
$$
which follows from the triangle inequality and a change of variables.

Let $z=(x,y)\in\R^3\times\R^3$. As $|x|^{-1}\Phi\in L^2(\R^6)$, we only need to show $\|\Delta_{h,z} (|x|^{-1}\Phi)\|_{L^2(\R^6)}\lesssim |h|^{1/2}$. By Sobolev's and Hardy's inequalities, we have
\begin{equation}\label{eq:2.23}
    \begin{split}
       &  \|\Delta_{h} (|x|^{-1}\Phi)\|_{L^2(\R^6)}\\
       \leq&  \|(\Delta_{h} |x|^{-1})\Phi\|_{L^2(\R^6)}+ \||x|^{-1}(\Delta_{-h} \Phi)\|_{L^2(\R^6)}\\
    \lesssim& \|\Delta_{h} |\cdot|^{-1}\|_{L^2(\R^3)}\|\Phi\|_{L^\infty(\R^3_x;L^2(\R^3_y))} +\|\Delta_{-h} \nabla_x \Phi\|_{L^2(\R^6)}\\
    \lesssim& |h|^{1/2}\|\Phi\|_{H^{2}(\R^6)}+|h|^{1/2}\|\nabla_x \Phi\|_{B^{1/2}_{2,\infty}(\R^6)}\lesssim |h|^{1/2}\|\Phi\|_{H^2(\R^6)}.
    \end{split}
\end{equation}

This shows that $|x|^{-1}\Phi\in B^{5/2}_{2,\infty}(\R^6)$. Analogously, $|y|^{-1}\Phi\in B^{5/2}_{2,\infty}(\R^6)$ and $|x-y|^{-1}\Phi\in B^{5/2}_{2,\infty}(\R^6)$ with the change of variable $r=\frac{1}{2}(x-y)$ and $R:=\frac{1}{2}(x+y)$. This proves \eqref{eq:1.55} and therefore $\Phi\in B^{5/2}_{2,\infty}(\R^6)$.


\subsection{Proof of $\Psi\in H^{3+\alpha}(\R^6)$}

Now we consider  $\Psi(x,y)=e^{-F(x,y)}\Phi(x,y)$ with $F$ the Jastrow factor in \eqref{eq:Jastrow}. We have
\begin{equation}\label{eq:inversexbesov}
\begin{split}
    \Delta_x\Phi=&e^{F}\Delta_x\Psi+\nabla_x e^{F} \cdot \nabla_x \Psi +\Psi\Delta_x e^F\\
    =& e^{F} \left( \Delta_x \Psi + \nabla_x F \cdot\nabla_x\Phi  + (|\nabla_x F|^2 + \Delta_x F) \Psi \right) \\
    =&e^{F}\Delta_x\Psi-e^F\left(  \frac{Zx}{|x|}\theta'(|x|)+\frac{x-y}{2|x-y|}\theta'(|x-y|)\right)\cdot \nabla_x\Psi\\
    &+e^F\left| \frac{Zx}{|x|}\theta'(|x|)+\frac{x-y}{2|x-y|}\theta'(|x-y|) \right|^2\Psi\\
    &-e^F\left(Z\theta''(|x|)+\frac{1}{2}\theta''(|x-y|) +\frac{2Z}{|x|}\theta'(|x|)+\frac{1}{|x-y|}\theta'(|x-y|)\right)\Psi.
\end{split}
\end{equation}

Inserting $\Phi(x,y)=e^{F(x,y)}\Psi(x,y)$ into \eqref{eq:two-body} yields
\begin{equation}\label{eq:euler-jastrow}
    -\frac{1}{2}\left(\Delta_x +\Delta_y\right)\Psi =-a_x\cdot\nabla_x \Psi-a_y\cdot\nabla_y \Psi-(b+c+d+2\mu)\Psi + \sum_{\lambda_j=1}e^{-F} e_j\psi_j(x)\psi_j^*(y)
\end{equation}
with
\begin{align}
    a_x(x,y)&:=  \frac{Zx}{2|x|}\theta'(|x|)+\frac{x-y}{4|x-y|}\theta'(|x-y|), \label{eq:ax}\\
    a_y(x,y)&:=  \frac{Zy}{2|y|}\theta'(|x|)-\frac{x-y}{4|x-y|}\theta'(|x-y|),\label{eq:ay}\\
    b(x,y)&:=-\frac{Z}{2}\left(\frac{x}{|x|}\theta'(|x|)-\frac{y}{|y|}\theta'(|y|)\right)\cdot\frac{x-y}{|x-y|}\theta'(|x-y|) \label{eq:b}\\
    c(x,y)&:=\rho_{\gamma}*|\cdot|^{-1}(x)+\rho_{\gamma}*|\cdot|^{-1}(y) \label{eq:c}
\end{align}
and
\begin{equation}\label{eq:d}
    \begin{split}
d(x,y)&:=-\frac{Z^2}{2}\big(\theta'(|x|)^2+\theta'(|y|)^2\big)-\frac{1}{4}\theta'(|x-y|)^2+\frac{Z}{2}\big(\theta''(|x|)+\theta''(|y|)\big)+\frac{1}{2}\theta''(|x-y|)\\
        &- \frac{Z}{|x|}(1-\theta'(|x|)) - \frac{Z}{|y|}(1-\theta'(|y|))- \frac{1}{|x-y|}(1-\theta'(|x-y|)).
    \end{split}
\end{equation}
We clearly have
\begin{align}
    a_x, a_y,b \in L^\infty(\R^6) \,.\label{eq:coeffbdd0}
\end{align}

Next, let us show that
\begin{align}
    c \in W^{1,\infty}(\R^6) \,.\label{eq:coeffbdd}
\end{align}
Indeed, the fact that $\rho_\gamma*|\cdot|^{-1}$ is bounded is an easy consequence of $\rho_\gamma\in L^1(\R^3)\cap L^3(\R^3)$, where the $L^3$-bound follows from the Sobolev and the Hoffmann--Ostenhof inequality. For the boundedness of the gradient of $\rho_\gamma*|\cdot|^{-1}$, we write
$$
\rho_\gamma(x)= \int_{\R^3} \Phi(x,y)\Phi(y,x)\,dy
$$
and use the embedding $H^1(\R^3)\subset L^\infty(\R^3)$ to deduce
$$
\| \rho_\gamma \|_{L^\infty(\R^3)}^2 \lesssim \|\Phi\|_{H^2(\R^3)}^2 \,.
$$
Thus,
\begin{equation}
    \label{eq:rholinfty}
    \rho_\gamma\in L^1(\R^3)\cap L^\infty(\R^3) \,.
\end{equation}
This easily implies that $\rho_\gamma * |\cdot|^{-2}$ is bounded and, consequently, that $\nabla(\rho_\gamma*|\cdot|^{-1})$ is bounded, as asserted in \eqref{eq:coeffbdd}.

Moreover, we have
\begin{equation}
    \label{eq:dsmoothbdd}
    d\in C^\infty_{\rm bdd}(\R^6) \,,
\end{equation}
where the subscript means that the function and all its derivatives are bounded. Indeed, the smoothness follows from $\theta\in C^\infty(\R)$, $\theta'(t)=1$ for $t\in [-1/2,1/2]$, and $\Supp(\theta')\subset [-1,1]$.

Using $\Phi\in H^2(\R^6)$, the definition of $F$ and Hardy's inequality, we have $\Psi\in H^2(\R^6)$.

\medskip

Next, we show that $\Psi\in H^3(\R^6)$. It suffices to show that each term on the right-hand side of \eqref{eq:euler-jastrow} is in $H^1(\R^6)$.

As $|\nabla_x a(x)|\lesssim |x|^{-1}$, by Hardy's inequality,
\begin{equation}
    \begin{split}
         \|a_x\cdot \nabla_x\Psi\|_{H^1(\R^6)}&\lesssim \||x|^{-1} \nabla\Psi\|_{H^1(\R^6)}+\|\Psi\|_{H^2(\R^6)}\lesssim \|\Psi\|_{H^2(\R^6)}.
    \end{split}
\end{equation}
Analogously,
\begin{equation}
      \|a_y\cdot \nabla_y\Psi\|_{H^1(\R^6)}\lesssim \|\Psi\|_{H^2(\R^6)}.
\end{equation}
As $$|\nabla_x b|\lesssim \frac{1}{|x|}+\frac{1}{|y|}+\frac{1}{|x-y|},$$by Hardy's inequality and the same argument,
\begin{equation}
      \|b \Psi\|_{H^1(\R^6)}\lesssim \|\Psi\|_{H^1(\R^6)} \leq \|\Psi\|_{H^2(\R^6)}.
\end{equation}
By \eqref{eq:hartree},
\begin{equation}
    \|c\Psi\|_{H^{1}(\R^6)}\leq \|c\|_{W^{1,\infty}(\R^6)}\|\Psi\|_{H^1(\R^6)}\lesssim \|\Phi\|_{H^1(\R^6)}\|\Psi\|_{H^2(\R^6)}.
\end{equation}
As $d\in C_{\rm bdd}^\infty(\R^6)$, we can immediately get
\begin{equation}
    \|d\Psi\|_{H^1(\R^6)}\lesssim \|\Psi\|_{H^2(\R^6)}.
\end{equation}
Finally, using $e^{-F}\lesssim 1$ and $|\nabla_{x} e^{-F}|+|\nabla_{y} e^{-F}|\lesssim 1$ and arguing similarly as in \eqref{eq:right-finite}, we find
\begin{equation}
    \left\|\sum_{\lambda_j=1}e^{-F} e_j\psi_j(x)\psi_j^*(y)\right\|_{H^1(\R^6)}\lesssim \sum_{\lambda_j=1} |e_j| \|\psi_j\|_{H^1(\R^3)} \lesssim \|\Phi\|_{H^1(\R^6)}.
\end{equation}
This shows that $\Psi\in H^3(\R^6)$.

\medskip

Now we show $\Psi\in H^{3+\alpha}(\R^6)$ for any $\alpha<\frac{1}{2}$. To do so, we are going to show that each terms on the right-hand side of \eqref{eq:euler-jastrow} is in $H^{1+\alpha}(\R^6)$.

We first study the term $a_x\cdot\nabla_x \Psi$. Indeed, as 
$$\left|\Delta \frac{x}{|x|}\right|\lesssim |x|^{-2}, \quad \left|\Delta \frac{x-y}{|x-y|}\right|\lesssim |x-y|^{-2},$$
analogous to \eqref{eq:1.13}, by Herbst's inequality for $\alpha<\frac{1}{2}$, we have
\begin{equation}
\begin{split}
  & \| a_x\cdot\nabla_x \Psi\|_{H^{1+\alpha}(\R^6)}\\
     =&\|(1-\Delta_x)^{(\alpha-1)/2} (1-\Delta_x) (a_x\cdot \nabla_x\Psi) \|_{L^2(\R^6)}\\
      \lesssim &\left\| |x|^{-1-\alpha}\nabla_x\Psi\right\|_{L^2(\R^6)}+\left\||x|^{-\alpha}\nabla_x\otimes \nabla_x\Psi\right\|_{L^2(\R^6)}+\left\|\nabla_x\Psi\right\|_{L^2(\R^6)}\\
     &+\left\| |x-y|^{-1-\alpha}\nabla_x\Psi\right\|_{L^2(\R^6)}+\left\||x-y|^{-\alpha}\nabla_x\otimes \nabla_x\Psi\right\|_{L^2(\R^6)}+\left\|\nabla_x\Psi\right\|_{L^2(\R^6)} \\
     \lesssim&\|\Psi\|_{H^3(\R^6)} .
\end{split}
\end{equation}
Analogously,
\begin{equation}
      \| a_y\cdot\nabla_y \Psi\|_{H^{1+\alpha}(\R^6)}\lesssim \|\Psi\|_{H^3(\R^6)}.
\end{equation}
As 
$$|\Delta_x b|\lesssim \frac{1}{|x|^2}+\frac{1}{|y|^2}+\frac{1}{|x-y|^2},$$
proceeding as above for $a_x\cdot\Psi$, we also have
\begin{equation}
     \| b \Psi\|_{H^{1+\alpha}(\R^6)}\lesssim \|\Psi\|_{H^3(\R^6)}.
\end{equation}
Next, using \eqref{eq:coeffbdd} and \eqref{eq:rholinfty} together with $-\Delta (\rho_\gamma*|\cdot|^{-1}) = 4\pi \rho_\gamma$, we find
\begin{equation}
    \begin{split}
         & \|c\Psi\|_{H^{1+\alpha}(\R^6)}\leq   \|c\Psi\|_{H^{2}(\R^6)}\\
         \lesssim& \|c\Psi\|_{L^2(\R^6)}+\|(\Delta_x c)\Psi\|_{L^2(\R^6)}+\|\nabla_x c\cdot\nabla_x\Psi \|_{L^2(\R^6)}+\|c (\Delta_x\Psi)\|_{L^2(\R^6)}\\
    \lesssim& (\|c\|_{W^{1,\infty}(\R^6)} + \|\rho_\gamma\|_{L^\infty(\R^3)}) \|\Psi\|_{H^2(\R^6)}\\
    & +\|(\Delta_y c)\Psi\|_{L^2(\R^6)}+\|\nabla_y c\cdot\nabla_x\Psi \|_{L^2(\R^6)}+\|c (\Delta_y\Psi)\|_{L^2(\R^6)}\\
    \lesssim & \|\psi\|_{H^2(\R^6)} \,.
    \end{split}
\end{equation}
In addition, as $d\in C_{\rm bdd}^\infty(\R^6)$,
\begin{equation}
    \|d\Psi\|_{H^{1+\alpha}(\R^6)} \leq \|d\Psi\|_{H^2(\R^6)} \lesssim \|\Psi\|_{H^2(\R^6)}.
\end{equation}

Finally, using 
$$|\Delta_x e^{-F}|+|\Delta_y e^{-F}|\lesssim \frac{1}{|x|}+\frac{1}{|y|}+\frac{1}{|x-y|}$$
and the $H^2$-analogue of Eq.~\eqref{eq:psi_j} we get 
\begin{equation}
    \begin{split}
         &\left\|\sum_{\lambda_j=1}e^{-F} e_j\psi_j(x)\psi_j^*(y)\right\|_{H^{1+\alpha}(\R^6)}\leq  \left\|\sum_{\lambda_j=1}e^{-F} e_j\psi_j(x)\psi_j^*(y)\right\|_{H^{2}(\R^6)}\\
         \lesssim &\sum_{\lambda_j=1}\|\psi_j\|_{H^2(\R^3)}^2 +\sum_{\lambda_j=1}\left\|\left(\frac{1}{|x|}+\frac{1}{|y|}+\frac{1}{|x-y|}\right)\psi_j(x)\psi_j^*(y)\right\|_{L^2(\R^6)}   \lesssim \|\Phi\|_{H^2(\R^6)}^2.
    \end{split}
\end{equation}

Gathering these estimates, we conclude that $\Psi\in H^{3+\alpha}(\R^6)$.


\subsection{Proof of $\Psi\in B^{7/2}_{2,\infty}(\R^6)$}\label{sec:2.4}
The proof is similar to the one of $\Phi\in B^{5/2}_{2,\infty}(\R^6)$. Similarly as in \eqref{eq:Delta-h-Delta}, we have
\begin{equation}
    \|\Phi\|_{B^{7/2}_{2,\infty}(\R^6)}\lesssim \|\Phi\|_{L^2(\R^6)}+ \sup_{h\neq 0} |h|^{-1/2}\|\Delta^{(1)}_{h}\Phi\|_{H^3(\R^6)} \,.
\end{equation}
Thus, it suffices to study $(-\Delta_x-\Delta_y)\nabla_x \Psi$ and $(-\Delta_x-\Delta_y)\nabla_y \Psi$. From \eqref{eq:euler-jastrow}, it suffices to show 
\begin{equation*}
   \nabla_x (a_x\cdot\nabla_x \Psi), \nabla_x(b\Psi) \in B^{1/2}_{2,\infty}(\R^6),
\end{equation*}
which follows by repeating the proof of \eqref{eq:2.23}. Other terms on the right-hand side \eqref{eq:euler-jastrow} are more regular, thus $(-\Delta_x-\Delta_y)\nabla_x \Psi$ and $(-\Delta_x-\Delta_y)\nabla_y \Psi$ are in $B^{1/2}_{2,\infty}(\R^6)$. As a result, $\Psi\in B^{7/2}_{2,\infty}(\R^6)$. The proof of Theorem \ref{th:regularity1} is complete. 


\section{Further regularity estimates}\label{sec:regularity2}

\subsection{Regularity estimates with explonential weights}

In this subsection we refine the regularity estimates from the previous section under the decay assumption that
\begin{equation}
    \label{eq:decayass}
    e^{\kappa|x|} \rho_{\gamma_*} \in L^1(\R^3)
\end{equation}
for some $\kappa>0$. We set
\begin{equation}
    \label{eq:defnu}
    \nu = \frac{\kappa}{4}
\end{equation}
and consider the functions
\begin{equation}
    \label{eq:phipsitilde}
    \widetilde\Phi(x,y) := e^{\nu(\langle x \rangle + \langle y \rangle)} \Phi(x,y) \,,
    \qquad
    \widetilde\Psi(x,y) := e^{\nu(\langle x \rangle + \langle y \rangle)} \Psi(x,y) \,.
\end{equation}

The following theorem is the analogue of Theorem \ref{th:regularity1}.

\begin{lemma}\label{lem:5.3}
    Let $\gamma_*$ be a minimizer for \eqref{eq:muellermin} in the molecular case and assume \eqref{eq:decayass} for some $\kappa>0$. Then we have
    \begin{equation*}
        \widetilde{\Phi}\in H^{2+\alpha}(\R^6:\C^{q\times q})\cap B^{5/2}_{2,\infty}(\R^6:\C^{q\times q}).
    \end{equation*}
    and
    \begin{equation*}
        \widetilde{\Psi}\in H^{3+\alpha}(\R^6:\C^{q\times q})\cap B^{7/2}_{2,\infty}(\R^6:\C^{q\times q}).
    \end{equation*}
    for any $0\leq \alpha<\frac{1}{2}$.
    As a consequence, $ \widetilde{\Phi}\in L^2(\R^3_x;B^{5/2}_{2,\infty}(\R^3_y:\C^{q\times q}))$ and $ \widetilde{\Psi}\in L^2(\R^3_x;B^{7/2}_{2,\infty}(\R^3_y:\C^{q\times q}))$.
\end{lemma}

\begin{proof}
    The proof is analogous to the one of Theorem \ref{th:regularity1}, so we only give a sketch of the argument and only consider the function $\widetilde{\Psi}$.

As $e^{\kappa |x|}\rho_{\gamma_*}(x)\in L^1(\R^3)$, we have $e^{\tfrac\kappa2 \left<x\right>}\rho_{\gamma_*}(x)\in L^1(\R^3)$, then by the Cauchy--Schwarz inequality
\begin{equation*}
    |\widetilde{\Psi}(x,y)|\lesssim e^{\tfrac\kappa4 \left<x\right>}\rho_{\gamma_*}(x)e^{\tfrac\kappa4 \left<y\right>}\rho_{\gamma_*}(y)\in L^2(\R^6).
\end{equation*}
From now on, let
\begin{equation*}
    g(x,y):=\sum_{\lambda_j=1}2e_j e^{\nu(\left<x\right>+\left<x\right>)}e^{-F(x,y)}\psi_j(y)\psi_j^*(y).
\end{equation*}
By the Cauchy--Schwarz inequality,
\begin{equation*}
    |g(x,y)|\leq  e^{\tfrac\kappa4 \left<x\right>}\rho_{\gamma_*}(x)e^{\tfrac\kappa4 \left<y\right>}\rho_{\gamma_*}(y)\in L^2(\R^6).
\end{equation*}

Multiplying Eq. \eqref{eq:euler-jastrow} by $e^{\nu(\left<x\right>+\left<y\right>)}$, we then have 
\begin{equation}\label{eq:euler-jastrow-exponential}
   -\frac{1}{2}\left(\Delta_x+\Delta_y\right)\widetilde{\Psi} = -\widetilde{a}_x \cdot \nabla_x \widetilde{\Psi}-\widetilde{a}_y\cdot \nabla_y \widetilde{\Psi}-(b+c+\widetilde{d}+2\mu)\widetilde{\Psi}+g
\end{equation}
where analogous to \eqref{eq:ax}-\eqref{eq:d}, from direct calculation, we know
\begin{equation*}
    \widetilde{a}_x-a_x, \widetilde{a}_y-a_y, \widetilde{d}-d \in C^\infty(\R^6:\C^{q\times q}).
\end{equation*}
Repeating the proof of Theorem \ref{th:regularity1} yields
$$\widetilde{\Psi}\in H^{3+\alpha}(\R^6:\C^{q\times q})\cap B^{7/2}_{2,\infty}(\R^6:\C^{q\times q})$$ with $0\leq \alpha<\frac{1}{2}$. In particular, we know $\widetilde{\Psi}\in L^2(\R_x^3;B^{7/2}_{2,\infty}(\R_y^3:\C^{q\times q}))$.
\end{proof}


\subsection{Local regularity estimates}

Let
\begin{equation}
    \label{eq:defxi}
    \Xi(x,y) := e^{-\frac12\theta(|x-y|)}\Psi(x,y) + \frac12|x-y| \Psi(x,x) \,.
\end{equation}

\begin{lemma}\label{lem:xi}
    Let $\chi\in C^\infty_c(\R^3)$. Then
    $$
    \chi(x) \Xi(x,y) \chi(y) \in L^2(\R_x^3;B^{7/2}_{2,\infty}(\R_y^3:\C^{q\times q})).
    $$
\end{lemma}

This lemma does not rely on the equation satisfied by $\Psi$. We only use $\Psi \in L^2(\R^3_x;B^{7/2}_{2,\infty}(\R^3_y:\C^{q\times q})$, as a consequence of Theorem \ref{th:regularity1}.

\begin{proof}
    Arguing as in Subsections \ref{sec:2.2} and \ref{sec:2.4}, it suffices to show that
    $$
    \nabla_y\Delta_y \left( \chi(x) \Xi(x,y) \chi(y) \right) \in L^2(\R_x^3;B^{1/2}_{2,\infty}(\R_y^3:\C^{q\times q})).
    $$
    The case where some derivatives fall on $\chi$ is easier to handle, so we concentrate on proving
    $$
    \chi(x) \left( \nabla_y \Delta_y \Xi(x,y) \right)  \chi(y)\in L^2(\R_x^3;B^{1/2}_{2,\infty}(\R_y^3:\C^{q\times q})).
    $$
    To prove this, we write $\Xi=\Xi_1 + \Xi_2$ with
    \begin{align*}
        \Xi_1(x,y) & := \left( e^{-\frac12\theta(|x-y|)} + \frac12|x-y|\right) \Psi(x,y) \,,\\
        \Xi_2(x,y) & := - \frac12|x-y| \left( \Psi(x,y) - \Psi(x,x) \right).
    \end{align*}
    Both $\Xi_1$ and $\Xi_2$ are products with two factors and we will only discuss the cases when all the derivatives fall on either one of the factors. The remaining cases can be treated by similar and simpler means.
    
    Concerning $\Xi_1$, the first case to consider is when all derivatives fall on $\Psi$. We bound in $L^2(\R^3_y)$ for fixed $x$
    \begin{align*}
        & \left\| \Delta_{h,y}\left[ \left( \chi(x) \left( e^{-\frac12\theta(|x-y|)} + \frac12|x-y| \right) \left( \nabla_y\Delta_y\Psi(x,y) \right) \chi(y) \right)\right] \right\|_{L^2(\R^3_y:\C^{q\times q})} \\
         \leq &\left\| \chi(x) \left( e^{-\frac12\theta(|x-y|)} + \frac12|x-y| \right) \Big[\Delta_{h,y} \left(\left( \nabla_y\Delta_y\Psi(x,y) \right) \chi(y) \right) \Big]\right\|_{L^2(\R^3_y:\C^{q\times q})} \\
        &  + \left\| \chi(x) \left[\Delta_{-h,y}  \left( e^{-\frac12\theta(|x-y|)} + \frac12|x-y| \right) \right]\Big[ \nabla_y\Delta_y\Psi(x,y) \chi(y) \Big]\right\|_{L^2(\R^3_y:\C^{q\times q})} \\
        \lesssim& |h|^{1/2} |\chi(x)| \|\Psi(x,y)\|_{B^{7/2}_{2,\infty}(\R^3_y:\C^{q\times q})}
    \end{align*}
    with a constant independent of $x$. Here we used the fact that
    $e^{-\frac12\theta(|x-y|)} + \frac12|x-y|$ and its derivative with respect to $y$ are bounded on compact subsets of $\R^3\times\R^3$.

    The second case to consider is when all derivatives fall on $e^{-\frac12\theta(|x-y|)} + \frac12|x-y|$. We bound, again in $L^2(\R^3_y)$ for fixed $x$
    \begin{align*}
        & \left\| \Delta_{h,y} \left( \chi(x) \nabla_y\Delta_y\left( e^{-\frac12\theta(|x-y|)} + \frac12|x-y| \right) \Psi(x,y) \chi(y) \right) \right\|_{L^2(\R^3_y:\C^{q\times q})} \\
        \leq& \left\| \chi(x) \left[\Delta_{h,y} \nabla_y\Delta_y\left( e^{-\frac12\theta(|x-y|)} + \frac12|x-y| \right) \right]\Psi(x,y) \chi(y) \right\|_{L^2(\R^3_y:\C^{q\times q})} \\
        & + \left\| \chi(x) \left[\nabla_y\Delta_y  \left( e^{-\frac12\theta(|x-y|)} + \frac12|x-y| \right)\right] \Big[\Delta_{-h,y} \left( \Psi(x,y) \chi(y) \right) \Big]\right\|_{L^2(\R^3_y:\C^{q\times q})} \\
         \lesssim& |h|^{1/2} |\chi(x)| \|\Psi(x,y)\|_{B^{7/2}_{2,\infty}(\R^3_y:\C^{q\times q})}
    \end{align*}
    with a constant independent of $x$. Here we used the fact that
    $$
    \nabla_y\Delta_y \left( e^{-\frac12\theta(|x-y|)} + \frac12|x-y| \right)
    $$
    is bounded on compact subsets of $\R^3\times\R^3$. Moreover, similarly as in Subsection \ref{sec:2.2}, we have
    $$
    \left\| \Delta_{h,y} \nabla_y\Delta_y\left( e^{-\frac12\theta(|x-y|)} + \frac12|x-y| \right) \right\|_{L^2(\R_y^3:\C^{q\times q})}
    \lesssim |h|^{1/2}
    $$
    uniformly in $x$. This proves the assertion for $\Xi_1$.

    Concerning $\Xi_2$, the case where all derivatives fall on $\Psi(x,y)-\Psi(x,x)$ is easy to handle. Thus, it remains to bound in $L^2(\R^3_y)$ for fixed $x$
    $$
    \left\| \Delta_{h,y} \left( \chi(x) \left( \nabla_y\Delta_y|x-y| \right) (\Psi(x,y)-\Psi(x,x)) \chi(y) \right) \right\|_2 \,.
    $$
    To do so, we write
    $$
    \Psi(x,y)-\Psi(x,x) = (y-x) \cdot G(x,y) \,,
    \qquad G(x,y) := \int_0^1 \nabla_2\psi(x,x+t(y-x))\,dt \,.
    $$
    We note that $\chi(x) G(x,y)\chi(y)\in L^2(\R^3_x;B^{5/2}_{2,\infty}(\R^3_y))$. We have to bound
    $$
    \left\| \Delta_{h,y} \left( \chi(x) k(x-y)\cdot G(x,y) \chi(y) \right) \right\|_2
    $$
    with $k(x-y) := (y-x) \nabla_y\Delta_y|x-y|$. Note that $|k(x-y)| \lesssim |x-y|^{-1}$ and similarly for the derivatives. The desired bound now follows by the same argument as in Subsection \ref{sec:2.2}, see, in particular, \eqref{eq:inversexbesov}, where the case of $|x|^{-1}$, rather than $k$ is handled. We omit the details. This completes the proof of Lemma \ref{lem:xi}.
\end{proof}

\section{Proof of Theorem \ref{th:decay}: Exponential decay}\label{sec:decay}

In this section we prove the decay estimate in Theorem \ref{th:decay} under the spectral gap condition $\mu<-1/2$, by following an argument of the first author that appears in \cite[Appendix B]{RooSei-22}. To simplify the notation, we assume $q=1$, but the argument works for general spin $q\in \N^+$. Also, we will write $\gamma$ instead of $\gamma_*$.

Given $R>0$, we choose two smooth, real-valued functions $\chi_<$ and $\chi_>$ on $\R^3$ such that
\begin{equation*}
    \Supp \chi_<\subset \overline{B_{2R}}\quad\mbox{and}\quad \Supp\chi_>\subset \R^3\setminus B_R
\end{equation*}
and such that $\chi_<^2+\chi_>^2=1$. We may assume that
\begin{equation*}
    |\nabla \chi_<|^2+|\nabla \chi_>|^2\leq CR^{-2}
\end{equation*}
with a constant $C$ independent of $R$. With $\kappa$ as in the theorem and with a parameter $\delta>0$ let 
$$
f_\delta(x)=\frac{\kappa |x|}{1+\delta |x|} \,.
$$
This function is bounded by $\kappa/\delta$.  Then, multiplying the two-body equation \eqref{eq:two-body} by $e^{2f_\delta(x)}\overline{\Phi(x,y)}$ and integrating, we have
\begin{equation}\label{eq:quadratic}
    \left<e^{2f_\delta}\Phi, (H_{\gamma}-2\mu)\Phi\right>_{L^2_{x,y}(\R^3\times \R^3)}=\sum_{\lambda_j=1}2e_j \left<e^{2f_\delta}\Phi, \psi_j(x)\psi_j^*(y)\right>_{L^2_{x,y}(\R^3\times \R^3)}.
\end{equation}
On the right-hand side of \eqref{eq:quadratic},
\begin{equation*}
    \sum_{\lambda_j=1}2e_j \left<e^{2f_\delta}\Phi, \psi_j(x)\psi_j^*(y)\right>_{L^2_{x,y}(\R^3\times \R^3)}=\sum_{\lambda_j=1}2e_j \int_{\R^3}e^{2f_\delta}|\psi_j(x)|^2dx.
\end{equation*}
On the left-hand side of \eqref{eq:quadratic}, we have
\begin{equation}\label{eq:2.4}
\begin{split}
     &  \left<e^{2f_\delta}\Phi, (H_{\gamma}-2\mu)\Phi\right>_{L^2_{x,y}(\R^3\times \R^3)}\\
   =&\left<e^{f_{\delta}}\Phi, \left(-\frac{1}{2}\Delta_x -\phi_\gamma(x)-|\nabla f_\delta|^2-\frac{1}{|x-y|}-\mu\right)e^{f_{\delta}}\Phi\right>\\
   &+\left<e^{f_{\delta}}\Phi, \left(-\frac{1}{2}\Delta_y -\phi_\gamma(y)-\mu\right)e^{f_{\delta}}\Phi\right>.
\end{split}
\end{equation}
Concerning the first term on the right-hand side of \eqref{eq:2.4}, by the IMS formula,
\begin{equation*}
\begin{split}
    &\left<e^{f_{\delta}}\Phi, \left(-\frac{1}{2}\Delta_x -\phi_\gamma(x)-|\nabla f_\delta|^2-\frac{1}{|x-y|}-\mu\right)e^{f_{\delta}}\Phi\right>\\
    =&\left<\chi_<(x) e^{f_{\delta}}\Phi, \left(-\frac{1}{2}\Delta_x -\phi_\gamma(x)-\frac{1}{|x-y|}\right)\chi_<(x)  e^{f_{\delta}}\Phi\right>\\
    &+\left<\chi_>(x) e^{f_{\delta}}\Phi, \left(-\frac{1}{2}\Delta_x -\phi_\gamma(x)-\frac{1}{|x-y|}\right)\chi_>(x)  e^{f_{\delta}}\Phi\right>\\
    &- \left<e^{f_{\delta}}\Phi, \left(|\nabla_x \chi_<(x)|^2+|\nabla_x \chi_>(x)|^2+|\nabla f_\delta|^2+\mu\right)e^{f_{\delta}}\Phi\right>.
\end{split}
\end{equation*}
Since $\Phi\in H^1(\R^6)$ and $|\nabla f_\delta|\leq 2\kappa$ independent of $\delta$,
\begin{equation*}
    \left|\left<\chi_<(x) e^{f_{\delta}}\Phi, \left(-\frac{1}{2}\Delta_x -\phi_\gamma(x)-\frac{1}{|x-y|}\right)\chi_<(x)  e^{f_{\delta}}\Phi\right>\right|\leq  C(R)
\end{equation*}
with a constant $C(R)$ that depends on $R$, but does not depend on $\delta$.

Next, since $\phi_\gamma$ decays like a constant times $|x|^{-1}$,
\begin{equation}\label{eq:2.5'}
\begin{split}
    & \left<\chi_>(x) e^{f_{\delta}}\Phi, \left(-\frac{1}{2}\Delta_x-\phi_{\gamma}(x) -\frac{1}{|x-y|}\right)\chi_>(x)  e^{f_\delta}\Phi\right>\notag\\
    \geq & \left<\chi_>(x) e^{f_{\delta}}\Phi, \left(-\frac{1}{2}\Delta_x -\frac{1}{|x-y|}\right)\chi_>(x)  e^{f_\delta}\Phi\right>-CR^{-1}\left\|\chi_> e^{f_\delta}\Phi\right\|_{L^2_{x,y}(\R^6)}^2 \notag\\
    \geq& -\left(\frac{1}{2}+CR^{-1}\right)\left\|\chi_> e^{f_\delta}\Phi\right\|_{L^2_{x,y}(\R^6)}^2.
\end{split}
\end{equation}
To summarize, the first term on the right-hand side of \eqref{eq:2.4} satisfies
\begin{equation}\label{eq:2.5}
\begin{split}
     &\left<e^{f_{\delta}}\Phi, \left(-\frac{1}{2}\Delta_x -\phi_\gamma(x)-|\nabla f_\delta|^2-\frac{1}{|x-y|}-\mu\right)e^{f_{\delta}}\Phi\right>\notag\\
     \geq& -(\|\nabla f_\delta\|_{L^\infty}^2 +CR^{-2}+CR^{-1}+\frac{1}{2}+\mu)\|e^{f_\delta}\Phi\|_{L^2(\R^6)}^2 -C(R) \,.
\end{split}
\end{equation}

For the second term on the right-hand side of \eqref{eq:2.4}, we have
\begin{equation}\label{eq:2.6}
\begin{split}
      &\left<e^{f_\delta}\Phi, \left(-\frac{1}{2}\Delta_y -\phi_\gamma(y)-\mu\right)e^{f_\delta}\Phi\right>\notag\\
    =&\left<e^{f_\delta}\Phi, \left(h_{\gamma,y}-\mu\right)e^{f_\delta}\Phi\right>+ \left<e^{f_\delta}\Phi, \mathfrak{X}_{\gamma,y} e^{f_\delta}\Phi\right>
\end{split}
\end{equation}
with the one-body operator $h_\gamma$ defined in \eqref{eq:onebodyh} and the nonlocal exchange operator $\mathfrak X_\gamma$ defined in \eqref{eq:xchangeop}. The subscript $y$ in $h_{\gamma,y}$ and $\mathfrak X_{\gamma,y}$ means that these operators act on the $y$ variable. By the one-body equation \eqref{eq:Muller-eigen}, we have
\begin{equation*}
\begin{split}
    h_{\gamma,y}\Phi(x,y)&=\sum_{j}\lambda_j^{1/2}\psi_j(x) h_{\gamma,y} \psi_j^*(y)\\
    &=\sum_{\lambda_j=1}(\mu+e_j)\psi_j(x) \psi_j^*(y)+\sum_{\lambda_j<1}\mu \psi_j(x) \psi_j^*(y)\\
    &=\mu \Phi(x,y)+\sum_{\lambda_j=1}e_j \psi_j(x)\psi_j^*(y).
\end{split}
\end{equation*}
Thus,
\begin{equation*}
\begin{split}
     &\left<e^{f_\delta}\Phi, \left(h_{\gamma,y} -\mu\right)e^{f_\delta}\Phi\right>\\
  =&\sum_{\lambda_j=1}e_j \left<e^{2f_\delta}\Phi, \psi_j(x)\psi_j^*(y)\right>_{L^2_{x,y}(\R^3\times \R^3)}=\sum_{\lambda_j=1}e_j \int_{\R^3}e^{2f_\delta(x)}|\psi_j(x)|^2dx.
\end{split}
\end{equation*}
This is exactly half of the term on the right-hand side of \eqref{eq:quadratic}. 

By the above equation and \eqref{eq:quadratic}--\eqref{eq:2.6}, we conclude that
\begin{equation*}
\begin{split}
     \MoveEqLeft -(\|\nabla f_\delta\|_{L^\infty}^2+CR^{-2}-CR^{-1}+\frac{1}{2}+\mu)\|e^{f_\delta}\Phi\|_{L^2(\R^6)}^2 \\
   &+ \left<e^{f_\delta}\Phi, \mathfrak{X}_{\gamma,y} e^{f_\delta}\Phi\right>-\sum_{\lambda_j=1}e_j \int_{\R^3}e^{2f_\delta}|\psi_j(x)|^2dx\leq C(R).
\end{split}
\end{equation*}
Since $\|\nabla f_\delta\|_{L^\infty}=\kappa$, since $\mathfrak{X}_{\gamma}$ is a non-negative operator and since $e_j\leq 0$, this implies
\begin{equation*}
\begin{split}
     \MoveEqLeft -(\kappa^2+CR^{-2}-CR^{-1}+\frac{1}{2}+\mu)\|e^{f_\delta}\Phi\|_{L^2(\R^6)}^2 \leq C(R).
\end{split}
\end{equation*}
Since $\kappa^2 <-\mu-1/2$, we can choose $R$ large (independently of $\delta$) such that the coefficient on the left side is positive. We deduce that
$$
\|e^{f_\delta}\Phi\|_{L^2(\R^6)}^2 \leq C'
$$
with $C'$ independent of $\delta$. We let $\delta\to 0$ and deduce, by monotone convergence, that
$$
\|e^{\kappa |x|}\Phi\|_{L^2(\R^6)}^2 \leq C' \,.
$$
Since
$$
\|e^{\kappa |x|}\Phi\|_{L^2(\R^6)}^2 = \int_{\R^3} e^{2\kappa |x|} \rho_\gamma(x)\,dx \,,
$$
the proof of Theorem \ref{th:decay} is complete. 


\section{Proof of Theorem \ref{th:mu}: Chemical potential} \label{sec:mu}

In this section we first derive the general upper bound \eqref{eq:mu-sigmaN} on the chemical potential $\mu$ in the molecular case, and then we make this bound more explicit in the atomic case when $N\le Z-C_0Z^{1/3}$ with a universal constant $C_0$. 

By the definition of $J$ in the statement of Theorem \ref{th:mu}, we can write $\gamma_* = \gamma_1 + \gamma_2$ with
\begin{equation*}
    \gamma_1 :=\sum_{j=1}^{J-1}\lambda_j\left|\psi_j\right>\left<\psi_j\right| \,,
    \qquad
    \gamma_2: = \sum_{j\geq J}\lambda_j\left|\psi_j\right>\left<\psi_j\right| \,.
\end{equation*}
Note that $\gamma_1$ is a projection since $\lambda_j=1$ with $j<J$, and $1>\gamma_2\ge 0$ since $1>\lambda_{J} \ge \lambda_{J+1}\ge \cdots \ge 0$. Let $u_1,\cdots, u_J$ be the first $J$ normalized eigenfunctions of $h$. Then there is a normalized function $u\in {\rm span}\{u_1,\ldots,u_J\}$ such that
\begin{equation*}
    \left<u,\psi_j\right>=0,\quad j=1,\cdots,J-1.
\end{equation*}
This implies that $\gamma_1 u=0$, and hence $u\in (1-\gamma_1)L^2(\R^3:\C^q)$. Note that
\begin{equation*}
    0\leq \gamma_2\leq \lambda_J(1-\gamma_1).
\end{equation*}
Then for $0\leq t\leq 1-\lambda_J$,
\begin{equation*}
0\leq  \gamma_2+t \left|u\right>\left<u\right|\leq 1-\gamma_1.
\end{equation*}
Thus $\gamma_2+t \left|u\right>\left<u\right|$ is a one-body density matrix for $0\leq t\leq 1-\lambda_J$, and so is 
$$
\widetilde{\gamma}_*(t):=\gamma_*+t\left|u\right>\left<u\right| \,.
$$
Note $\tr[\widetilde{\gamma}_*(t)]=N+t$. Now,
\begin{equation*}
    \mu=\lim_{t\to 0^+}\frac{E^{\rm M}(N+t)-E^{\rm M}(N)}{t}\leq \liminf_{t\to 0}\frac{\mathcal{E}^{\rm M}(\widetilde{\gamma}_*(t))-\mathcal{E}^{\rm M}(\gamma_*)}{t} \,.
\end{equation*}
Note that as $\gamma\mapsto X(\gamma^{1/2})$ is homogeneous and concave, as shown in \cite{FraLieSeiSie-07}, based on \cite{WignerYanase1963,WignerYanase1964}. As a consequence,
\begin{equation*}
    -X(\widetilde{\gamma}_*^{1/2}(t))\leq -(1-t)X((1-t)^{-1/2}\gamma_*^{1/2})-tX(\left|u\right>\left<u\right|)=X(\gamma_*^{1/2})-tX(\left|u\right>\left<u\right|).
\end{equation*}
Moreover, 
\begin{align*}
    X[\left|u\right>\left<u\right|]&=\frac{1}{2}\int_{\R^3}\int_{\R^3} \frac{|u|^2(x)|u|^2(y)}{|x-y|}dxdy\\
    &\geq \frac{1}{2}\inf_{\substack{\|u\|_{L^2}=1\\ u\in {\rm Span}\{u_1,\cdots,u_N\}}}\int_{\R^3} \int_{\R^3}\frac{|u|^2(x)|u|^2(y)}{|x-y|}dxdy.
\end{align*}
Thus we obtain the desired bound \eqref{eq:mu-sigmaN} as  
\begin{align*}
     \mu&\leq \left<u, (k_{\gamma_*}\otimes 1_{\C^q}) u\right>-X[\left|u\right>\left<u\right|] \\
     &\leq \sigma_J(h)-\frac{1}{2}\inf_{\substack{\|u\|_{L^2}=1\\ u\in {\rm Span}\{u_,1_,\cdots,u_N\}}}\int_{\R^3} \int_{\R^3} \frac{|u|^2(x)|u|^2(y)}{|x-y|}dxdy<\sigma_J(h).
\end{align*}
This ends the first step. 

In the second step, let us focus on the atomic case $V(x)=Z|x|^{-1}$ with $N\le Z-C_0Z^{1/3}$ for a constant $C_0>0$ and $Z$ sufficiently large. 

Note that $\rho_{\gamma_*}$ is radial in the atomic case \cite{FraLieSeiSie-07}. By Newton's theorem,
\begin{equation*}
    k_{\gamma_*}\leq -\frac{1}{2}\Delta -\frac{Z-N}{|x|}=: k.
\end{equation*}

We apply the bound from the first part of Theorem \ref{th:mu} with $h:=k\otimes1_{\C^q}$. First, assume that $N \ge 1$ is an integer. Then, since $\gamma_*$ has infinitely many positive eigenvalues, we have $\lambda_N<1$ and therefore $J\leq N$. Therefore, it suffices to show that $\sigma_N(h)\leq -\frac12$ to obtain the bound $\mu<-\frac12$.

Recall that $h$ has eigenvalues $-\frac{(Z-N)^2}{2n^2}$ with multiplicity $q n^2$, for $n=1,2,...$. Therefore,  
\begin{equation*}
    \sigma_N(h)=-\frac{(Z-N)^2}{2 n_0^2}
\end{equation*}
where $n_0$ is determined by
\begin{equation}\label{eq:5.1}
  \frac{n_0(n_0-1)(2n_0-1)}{6}q =\sum_{m=1}^{n_0-1}qm^2 < N\leq  \sum_{m=1}^{n_0}q m^2=\frac{n_0(n_0+1)(2n_0+1)}{6}q.
\end{equation}
Using 
\begin{align*}
    \frac{\lceil (3q^{-1}N)^{1/3}\rceil(\lceil (3q^{-1}N)^{1/3}\rceil+1)(2\lceil (3q^{-1}N)^{1/3}\rceil+1)}{6}q\geq \frac{2\lceil (3q^{-1}N)^{1/3}\rceil^3q}{6}\geq N \,,
\end{align*}
we deduce from \eqref{eq:5.1} that $n_0\leq  \lceil (3q^{-1}N)^{1/3}\rceil$. 

Moreover, we have
\begin{equation*}
    \sigma_N(h)\leq -\frac{(Z-N)^2}{2  \lceil (3q^{-1}N)^{1/3}\rceil^2}\leq -\frac{(Z-N)^2}{2((3q^{-1}N)^{1/3}+1)^2}.
\end{equation*}
Under the condition $N\leq Z-C_0 Z^{1/3} $ with $C_0>(3/q)^{1/3}$, we deduce from \eqref{eq:mu-sigmaN} that   
\begin{equation*}
\begin{split}
      \mu <  \sigma_N(h)\leq& -\frac{C_0^2 Z^{2/3}}{2((3/q)^{1/3} (Z-C_0Z^{1/3})^{1/3}+1)^2} = -\frac{C_0^2}{2\sqrt[3]{3q}}+O(Z^{-1/3})<-\frac{1}{2}
\end{split}
\end{equation*}
for $Z$ large enough.

When $N$ is not an integer, we have $\lambda_{[N]+1}<1$ and therefore $J\leq[N]+1$. A similar computation as above shows that even in this case $\sigma_{[N]+1}(h)\leq-\frac12$ for $N\leq Z-C_0 Z^{1/3}$ and $Z$ sufficiently large.

The proof of Theorem \ref{th:mu} is complete. 


\section{Proof of Theorem \ref{th:asymptotic}: Eigenvalue asymptotics}\label{sec:main-thm}


\subsection{Some background on Schatten classes}\label{sec:schatten}
 
Let us recall some results on compact operators. Let $T$ be a compact operator in a Hilbert space. If $T=T^*\geq 0$, let $\lambda_k(T),\;k=1,2,\ldots$, be the positive eigenvalues of $T$ numbered in descending order, counting multiplicity. In general, we define $s_k(T)=\lambda_k(\sqrt{TT^*})=\lambda_k(\sqrt{T^*T})$.

We say $T\in \mathfrak{S}_{p,\infty}$ if
\begin{equation*}
    \|T\|_{p,\infty}:=\sup_{k} k^{\frac{1}{p}} s_k(T)<\infty.
\end{equation*}
For all $p>0$, we have the ``triangle'' inequality
\begin{equation}
    \|T_1+T_2\|^{\frac{p}{p+1}}_{p,\infty}\leq \|T_1\|_{p,\infty}^{\frac{p}{p+1}}+\|T_2\|_{p,\infty}^{\frac{p}{p+1}}
\end{equation}
and for $0<p<1$
\begin{equation}\label{eq:5.4}
    \left\|\sum_j T_j\right\|_{p,\infty}^{p}\leq (1-p)^{-1}\sum_{j}\|T_j\|_{p,\infty}^p
\end{equation}
see \cite[Section 1]{BirSol-77} or \cite[Section 11.6]{BirSol-12} and references therein.

We also define
\begin{equation*}
    \mathfrak{G}_p(T):=\left(\limsup_{k\to\infty} k^{1/p}s_k(T)\right)^p,\qquad  \mathfrak{g}_p(T):=\left(\liminf_{k\to\infty} k^{1/p}s_k(T)\right)^p.
\end{equation*}
It is clear that
\begin{equation}\label{eq:Gp-Sp}
    \mathfrak{g}_p(T)\leq   \mathfrak{G}_p(T)\leq \|T\|_{p,\infty}^p.
\end{equation}
In addition,
\begin{equation}\label{eq:T1-T2-G}
    \begin{split}
     \left|\mathfrak{G}_p(T_1)^{\frac{1}{p+1}}-\mathfrak{G}_p(T_2)^{\frac{1}{p+1}}\right|\leq \mathfrak{G}_p(T_1-T_2)^{\frac{1}{p+1}},\\
     \left|\mathfrak{g}_p(T_1)^{\frac{1}{p+1}}-\mathfrak{g}_p(T_2)^{\frac{1}{p+1}}\right|\leq \mathfrak{G}_p(T_1-T_2)^{\frac{1}{p+1}}.
\end{split}
\end{equation}

We will need the following seminal result of Birman and Solomyak \cite[Section 4]{BirSol-77} (see also \cite[Proposition 3.3]{Sobolev22_kinetic}) concerning Schatten class properties of a specific class of integral operators.

\begin{theorem}\label{th:operator-norm}
    Let $\Omega\subset \R^d$ be a bounded Lipschitz domain, let $q,\ell\in\N^*$ and $s>0$. Assume that $T(t,\cdot)\in B^{s}_{2,\infty}(\Omega:\C^{q\times q})$ for a.e.~$t\in \R^\ell$, and assume that $b\in L^2_{\rm loc}(\R^\ell,\C^{q\times q})$ and $a\in L^{r}(\Omega:\C^{q\times q})$ with
   \begin{equation}
       \begin{cases}
        r=2,\qquad &\text{if}\quad 2s>d,\\
        r>ds^{-1},\qquad&\text{if}\quad 2s\leq d.
    \end{cases}
   \end{equation} 
    Then $bTa\in \mathfrak{S}_{p,\infty}$ with $p^{-1}=2^{-1}+sd^{-1}$ and
    \begin{equation}
        \|bTa\|_{\mathfrak{S}_{p,\infty}}\lesssim \left[\int_{\R^\ell}\|T(t,\cdot)\|_{B^{s}_{2,\infty}(\Omega:\C^{q\times q})}^2 |b(t)|^2dt\right]^{1/2} \|a\|_{L^r(\Omega:\C^{q\times q})},
    \end{equation}
    provided the right side is finite. The implicit constant depends only on $\Omega$, $s$, $r$, $d$, $\ell$ and $q$.
\end{theorem}

Concerning a special homogeneous case, we can use the following result from \cite{Birman_Solomyak_asymphomo} (see also \cite[Example 3.7]{sobolev2022eigenvalue2}).

\begin{theorem}\label{th:operator-norm-homo}
    Let $A,B\in C_c(\R^3:\C^{q\times q})$ and let $T$ be the operator with the kernel $T(x,y)=A(x)|x-y|B(y)$. Then
 \begin{equation}
     \mathfrak{G}_{3/4}(T)=\mathfrak{g}_{3/4}(T)= \frac{1}{3}\left(\frac{2}{\pi}\right)^{5/4}\int_{\R^3}|\Tr_{\C^q}(A(x)B(x))|^{3/4} \,dx.
 \end{equation}
\end{theorem}


\subsection{Conclusion}

In this section we complete the proof of the main result in Theorem \ref{th:asymptotic}. In the following we will assume that $V(x)=\frac{Z}{|x|}$. Under the assumption that $Z$ is sufficiently large and that $N\leq Z - C_0 Z^{1/3}$, we know from Theorems \ref{th:decay} and \ref{th:mu} that $e^{\kappa |x|}\rho_{\gamma_*}\in L^1(\R^3)$ for some $\kappa>0$. Consequently, we can make use of the results in Section \ref{sec:regularity2}.

We emphasize that the restriction to the atomic case, as well as the assumptions on $Z$ and $N$ in Theorem \ref{th:asymptotic} are only used to ensure that $e^{\kappa |x|}\rho_{\gamma_*}\in L^1(\R^3)$ for some $\kappa>0$. If this can be proved by other means, then the proof in the remainder of this section remains valid and consequently one arrives at the conclusion of Theorem \ref{th:asymptotic}.

We define the operator $T_K$ generated by $K\in L^2(\R^6)$ as follows: for any $f\in L^2(\R^3)$,
\begin{equation*}
    T_K f(x):=\int_{\R^3} K(x,y)f(y)dy.
\end{equation*}

We divide the proof of Theorem \ref{th:asymptotic} into three steps.

\medskip

\emph{Step 1.} In this step we show that
\begin{equation}
    \label{eq:ordersharp}
    \gamma_*^{1/2}\in \mathfrak{S}_{3/4,\infty} \,.
\end{equation}
Moreover, let $\eta\in C^\infty_c(\R,[0,1])$ satisfy 
\begin{equation*}
     \eta(t)=1 \quad{\rm for } \quad t\in [-1/2,1/2] \qquad {\rm and} \qquad \eta=0 \quad{\rm for} \quad|t|\geq 1 \,,
\end{equation*}
and set, for $x\in\R^3$,
$$
\eta_R(x):=\eta(|x|/R) \,.
$$
We will also show that, as $R\to \infty$, we have
\begin{equation}\label{eq:ordersharpcutoff}
     \|\gamma_*^{1/2}-\eta_R\gamma_*^{1/2}\eta_R\|_{3/4,\infty}\to 0.
\end{equation}

For the proof of \eqref{eq:ordersharp} we let $\mathcal{C}_n:=[-\frac12,\frac12)^3+n,\;n\in\Z^3$. Recalling $\Phi(x,y)=\gamma_*^{1/2}(x,y)$ we set, as in Section \ref{sec:regularity2},
$$
\widetilde\Phi(x,y) := e^{\nu(\langle x \rangle + \langle y \rangle)} \Phi(x,y)
\qquad\text{with}\ \nu:= \frac{\kappa}{4} \,.
$$
and note that
\begin{equation*}
    \gamma_*^{1/2}=e^{-\nu\left<\cdot\right>}T_{\widetilde{\Phi}}e^{-\nu\left<\cdot\right>} \,.
\end{equation*}
Thus, we deduce from Theorem \ref{th:operator-norm} with $s=\frac 52$ and $\Omega = \mathcal C_n$ that
\begin{equation*}
    \begin{split}
        \|\gamma_*^{1/2}\mathbbm{1}_{\mathcal{C}_n}\|_{3/4,\infty}&\lesssim \left[\int_{\R^3} e^{-2\nu\left<x\right>}\|\widetilde{\Phi}(x,\cdot)\|_{B^{5/2}_{2,\infty}(\R^3:\C^{q\times q})}^2 \,  dx\right]^{1/2} \|e^{-\nu\left<\cdot \right>}\|_{L^2(\mathcal{C}_n)}\\
        &\lesssim \|\widetilde{\Phi}\|_{L^2(\R^3_x;B^{5/2}_{2,\infty}(\R^3_y:\C^{q\times q}))}e^{-\nu|n|_1}.
    \end{split}
\end{equation*}
(Clearly, the implicit constant here is independent of $n$.)
According to Lemma \ref{lem:5.3} the norm on the right side is finite. Here and below, for any vector $a=(a_1,\cdots,a_d)\in \R^d$, we define the $|\cdot|_1$ norm as 
 \begin{equation*}
     |a|_1=\sum_{j=1}^d |a_j|.
 \end{equation*}
This and \eqref{eq:5.4} show that
\begin{equation}\label{eq:T-Cn-sum}
    \|\gamma_*^{1/2}\|_{3/4,\infty}^{3/4}=\left\|\sum_n \gamma_*^{1/2}\mathbbm{1}_{\mathcal{C}_n} \right\|_{3/4,\infty}^{3/4}\lesssim \sum_{n}e^{-\nu|n|_1}<\infty,
\end{equation}
thus proving \eqref{eq:ordersharp}.

The proof of \eqref{eq:ordersharpcutoff} is similar. We have
\begin{equation*}
    \|\gamma_*^{1/2}-\eta_R\gamma_*^{1/2}\eta_R\|_{3/4,\infty}\lesssim  \|(1-\eta_R)\gamma_*^{1/2}\eta_R\|_{3/4,\infty}+\|\gamma_*^{1/2}(1-\eta_R)\|_{3/4,\infty}.
\end{equation*}
From Theorem \ref{th:operator-norm} again, 
\begin{align*}
    \|(1-\eta_R)\gamma_*^{1/2} \eta_R \mathbbm{1}_{\mathcal{C}_n} \|_{3/4,\infty}&\lesssim  \left[\int_{|x|\geq \frac{R}{2}} e^{-2\nu\left<x\right>}\|\widetilde{\Phi}(x,\cdot)\|_{B^{5/2}_{2,\infty}(\R^3:\C^{q\times q})}  dx\right]^{1/2} \|\eta_R e^{-\nu\left<\cdot\right>}\|_{L^2(\mathcal C_n)} \\
    &\lesssim e^{-\frac12\nu R} e^{-\nu |n|_1}
\end{align*}
and
\begin{align*}
      \|\gamma_*^{1/2}(1-\eta_R)\mathbbm{1}_{\mathcal{C}_n} \|_{3/4,\infty} & \lesssim  \left[\int_{\R^3} e^{-2\nu\left<x\right>}\|\widetilde{\Phi}(x,\cdot)\|_{B^{5/2}_{2,\infty}(\R^3:\C^{q\times q})}  dx\right]^{1/2}\|(1-\eta_R) e^{-\nu\left<x\right>}\|_{L^2(\mathcal C_n)} \\
      & \lesssim e^{-\nu |n|_1} \mathbbm 1(|n|_1 \geq cR) \,.
\end{align*}
Thus, using \eqref{eq:5.4} and arguing as in \eqref{eq:T-Cn-sum}, we arrive at
$$
\|\gamma_*^{1/2}-\eta_R\gamma_*^{1/2}\eta_R\|_{3/4,\infty}^{3/4} \lesssim
e^{-\frac12\nu R} \sum_n e^{-\nu |n|_1} + \sum_{|n|_1\geq cR} e^{-\nu |n|_1} \,.
$$
Since the right side tends to zero as $R\to 0$, we arrive at \eqref{eq:ordersharpcutoff}.

\medskip

\emph{Step 2.} Next, setting
\begin{equation*}
    \Theta(x,y) :=-\frac{1}{2}|x-y| \, \Psi(x,x)
\end{equation*}
with $\Psi$ defined in Theorem \ref{th:regularity1}, we show in this step that, for each $R>0$ fixed,
\begin{equation}
    \label{eq:ordersharpremainder}
    \eta_R\Big(\gamma_*^{1/2}-e^{-Z\theta(|\cdot|)}T_{\Theta}e^{-Z\theta(|\cdot|)}\Big)\eta_R \in \mathfrak S_{3/5,\infty} \,.
\end{equation}
Here, $\theta$ is the function from \eqref{eq:theta}.

To prove \eqref{eq:ordersharpremainder}, we recall the definition of $\Xi$ in \eqref{eq:defxi} and decompose
\begin{align*}
    \gamma_*^{1/2}(x,y) & = 
    e^{-Z\theta(|x|)} \Xi(x,y) e^{-Z\theta(|y|)}
    + e^{-Z\theta(|x|)} \Theta(x,y) e^{-Z\theta(|y|)} \,.
\end{align*}

We apply Theorem \ref{th:operator-norm} with $s=\frac 72$ and $\Omega = B_R$ (recall that $R$ is fixed in this step) and obtain
\begin{align*}
    & \left\| \eta_R \left( \gamma_*^{1/2} - e^{-Z\theta(|\cdot|)}T_{\Theta}e^{-Z\theta(|\cdot|)}  \right) \eta_R \right\|_{3/5,\infty} \\
    & = \left\| 1_{B_R} e^{-Z\theta(|\cdot|)} T_{\eta_R \Xi \eta_R} e^{-Z\theta(|\cdot|)} 1_{B_R} \right\|_{3/5,\infty} \\
    & \lesssim \left[ \int_{B_R} e^{-2Z\theta(|x|)} \| \eta_R(x) \Xi(x,\cdot) \eta_R \|_{B^{7/2}_{2,\infty}(\R^3:\C^{q\times q})}^2 \,dx \right]^{1/2} \left\| e^{-Z\theta(|\cdot|)} \right\|_{L^2(B_R)}.
\end{align*}
The second factor on the right side is clearly finite. In the first factor, we bound $e^{-2Z\theta(|x|)}\leq 1$ and use the fact that, by Lemma \ref{lem:xi},
$$
\eta_R(x)\Xi(x,y)\eta_R(y) \in L^2(\R^3_x; B^{7/2}_{2,\infty}(\R^3_y:\C^{q\times q}) \,.
$$
This completes the proof of \eqref{eq:ordersharpremainder}.

\medskip

\emph{Step 3.} We can finally complete the proof of our main theorem.

We shall make use of the asymptotic functionals $\mathfrak{G}_{3/4}$ and $\mathfrak{g}_{3/4}$ introduced in Subsection \ref{sec:schatten}. By \eqref{eq:T1-T2-G}, we have, for any $R>0$,
\begin{equation}\label{eq:mainproof1a}
    \begin{split}
    \mathfrak{G}_{3/4}(\gamma_*^{1/2})^{\frac{4}{7}}
    & \leq \mathfrak{G}_{3/4}(\gamma_*^{1/2}-\eta_R \gamma_*^{1/2}\eta_R)^{\frac{4}{7}} \\
    & \quad + \mathfrak{G}_{3/4}\left(\eta_R\Big(\gamma_*^{1/2}-   e^{-Z\theta(|\cdot|)}T_{\Theta}e^{-Z\theta(|\cdot|)}\Big) \eta_R\right)^{\frac{4}{7}}\\
   & \quad + \mathfrak{G}_{3/4}\left(\eta_R e^{-Z\theta(|\cdot|)} T_{\Theta}e^{-Z\theta(|\cdot|)} \eta_R\right)^{\frac{4}{7}}
   \end{split}
\end{equation}
and
\begin{equation}\label{eq:mainproof1b}
    \begin{split}
    \mathfrak{g}_{3/4}(\gamma_*^{1/2})^{\frac{4}{7}}
    & \geq - \mathfrak{G}_{3/4}(\gamma_*^{1/2}-\eta_R \gamma_*^{1/2}\eta_R)^{\frac{4}{7}} \\
    & \quad - \mathfrak{G}_{3/4}\left(\eta_R\Big(\gamma_*^{1/2}-   e^{-Z\theta(|\cdot|)}T_{\Theta}e^{-Z\theta(|\cdot|)}\Big) \eta_R\right)^{\frac{4}{7}}\\
   & \quad + \mathfrak{g}_{3/4}\left(\eta_R e^{-Z\theta(|\cdot|)} T_{\Theta}e^{-Z\theta(|\cdot|)} \eta_R\right)^{\frac{4}{7}}.
    \end{split}
\end{equation}

By Theorem \ref{th:operator-norm-homo}, we have
\begin{equation}\label{G-Theta-R}
    \begin{split}
       & \mathfrak{g}_{3/4}(\eta_R e^{-Z\theta(|\cdot|)}T_{\Theta } e^{-Z\theta(|\cdot|)} \eta_R)\\
       =& \mathfrak{G}_{3/4}(\eta_R e^{-Z\theta(|\cdot|)}T_{\Theta } e^{-Z\theta(|\cdot|)} \eta_R) \\
        =& \frac{1}{3}\left(\frac{2}{\pi}\right)^{5/4}\int_{\R^3} \left(\frac12\eta_R^2 e^{-2Z\theta(|x|)}\Tr_{\C^q}\Psi(x,x)\right)^{3/4}dx\\
        =& \frac{1}{3}\left(\frac{2}{\pi}\right)^{5/4}\int_{\R^3} \left(\frac12\eta_R^2 \Tr_{\C^q}\Phi(x,x)\right)^{3/4}dx\\
        =& \frac{1}{3}\left(\frac{2}{\pi}\right)^{5/4}\int_{\R^3} \left(\frac12\eta_R^2 \rho_{\gamma_*}\right)^{3/4}dx.
    \end{split}
\end{equation}
Note that this theorem is applicable since $x\mapsto\Psi(x,x)$ is continuous as a consequence of Theorem \ref{th:regularity1} and the Sobolev embedding theorem.

Moreover, by \eqref{eq:ordersharpremainder} and since $3/5<3/4$, we have
\begin{equation}
    \label{eq:proofmainstep2}
        \mathfrak{G}_{3/4}\left(\eta_R\Big(\gamma_*^{1/2}-e^{-Z\theta(|\cdot|)}T_{\Theta}e^{-Z\theta(|\cdot|)}\Big)\eta_R\right)=0.
\end{equation}
Thus, \eqref{eq:mainproof1a} and \eqref{eq:mainproof1b} can be rewritten as
\begin{equation}\label{eq:mainproof2a}
    \begin{split}
    \mathfrak{G}_{3/4}(\gamma_*^{1/2})^{\frac{4}{7}}
    & \leq \mathfrak{G}_{3/4}(\gamma_*^{1/2}-\eta_R \gamma_*^{1/2}\eta_R)^{\frac{4}{7}} + \left( \frac{1}{3}\left(\frac{2}{\pi}\right)^{5/4}\int_{\R^3} \left(\frac12\eta_R^2 \rho_{\gamma_*}\right)^{3/4}dx \right)^\frac47
   \end{split}
\end{equation}
and
\begin{equation}\label{eq:mainproof2b}
    \begin{split}
    \mathfrak{g}_{3/4}(\gamma_*^{1/2})^{\frac{4}{7}}
    & \geq - \mathfrak{G}_{3/4}(\gamma_*^{1/2}-\eta_R \gamma_*^{1/2}\eta_R)^{\frac{4}{7}} + \left( \frac{1}{3}\left(\frac{2}{\pi}\right)^{5/4}\int_{\R^3} \left(\frac12\eta_R^2 \rho_{\gamma_*}\right)^{3/4}dx \right)^\frac47.
    \end{split}
\end{equation}
These bounds are valid for arbitrary $R>0$. Now, we take the limit $R\to\infty$. Noting that \eqref{eq:ordersharpcutoff} implies, by \eqref{eq:Gp-Sp}, 
$$
\lim_{R\to\infty} \mathfrak{G}_{3/4}(\gamma_*^{1/2}-\eta_R \gamma_*^{1/2}\eta_R) = 0
$$
and using
$$
\lim_{R\to\infty} \int_{\R^3} \left(\frac12\eta_R^2 \rho_{\gamma_*}\right)^{3/4}dx = \int_{\R^3} \left(\frac12 \rho_{\gamma_*}\right)^{3/4}dx \,,
$$
we finally arrive at
$$
\mathfrak{G}_{3/4}(\gamma_*^{1/2}) = \mathfrak{g}_{3/4}(\gamma_*^{1/2}) = \frac{1}{3}\left(\frac{2}{\pi}\right)^{5/4}\int_{\R^3} \left(\frac12 \rho_{\gamma_*}\right)^{3/4}dx \,,
$$
which is the assertion of the theorem. This completes the proof.

\end{document}